\documentclass[10pt,journal,final]{IEEEtran}
\usepackage[utf8]{inputenc}
\usepackage{cite}
\usepackage{CJK}
\usepackage{tipa}
\usepackage{amsmath}
\usepackage{enumerate}
\usepackage{graphicx}
\usepackage{multirow}
\usepackage{array}
\usepackage{stfloats}
\usepackage{multicol}
\usepackage{amsthm,amsmath}
\usepackage{mathrsfs}
\usepackage[normalem]{ulem} 
\usepackage{booktabs}
\usepackage{stmaryrd}
\usepackage[ruled]{algorithm2e}
\usepackage{color}
\usepackage{soul}
\usepackage{subfigure}
\usepackage{makecell}
\usepackage{threeparttable}
\usepackage{amssymb}
\usepackage{tikz,xcolor}
\usepackage[implicit=false]{hyperref}

\hypersetup{hidelinks,
	colorlinks=true,
	allcolors=black,
	pdfstartview=Fit,
	breaklinks=true}

\definecolor{lime}{HTML}{A6CE39}
\DeclareRobustCommand{\orcidicon}{
\begin{tikzpicture}
\draw[lime, fill=lime] (0,0)
circle[radius=0.16]
node[white]{{\fontfamily{qag}\selectfont \tiny \.{I}D}};
\end{tikzpicture}
\hspace{-2mm}
}
\foreach \x in {A, ..., Z}{%
\expandafter\xdef\csname orcid\x\endcsname{\noexpand\href{https://orcid.org/\csname orcidauthor\x\endcsname}{\noexpand\orcidicon}}
}
\makeatletter

\newcommand{\Rmnum}[1]{\expandafter\@slowromancap\romannumeral #1@}
\makeatother

\begin{document}
\title{FairCMS: Cloud Media Sharing with Fair Copyright Protection }
\author
{\IEEEauthorblockN{Xiangli Xiao}\hspace{-1.5mm}\orcidA{},
        {Yushu Zhang}\hspace{-1.5mm}\orcidB{},
        {Leo Yu Zhang}\hspace{-1.5mm}\orcidC{}, 
        {Zhongyun Hua}\hspace{-1.5mm}\orcidD{},
        {Zhe Liu}\hspace{-1.5mm}\orcidE{},
        and
        {Jiwu Huang}\hspace{-1.5mm}\orcidF{}

\thanks{This work has been published in IEEE Transactions on Computational Social Systems, DOI: 10.1109/TCSS.2024.3374452.}
\thanks{This work was supported in part by the National Key R\&D Program of China under Grant 2021YFB3100400 and in part by the Postgraduate Research \& Practice Innovation Program of Jiangsu Province under Grant KYCX23\_0397. \textit{(Corresponding author: Yushu Zhang and Leo Yu Zhang.)}} 
\thanks{X. Xiao, Y. Zhang, and Z. Liu are with the College of Computer Science and Technology, Nanjing University of Aeronautics and Astronautics, Nanjing 211106, China (e-mail: xiaoxiangli@nuaa.edu.cn; yushu@nuaa.edu.cn; zhe.liu@nuaa.edu.cn). }
\thanks{L. Zhang is with the School of Information and Communication Technology, Griffith University, Southport QLD 4215, Australia (e-mail: leocityu@gmail.com). }
\thanks{Z. Hua is with the School of Computer Science and Technology, Shenzhen Campus of Harbin Institute of Technology, Shenzhen 518055, China (e-mail: huazhongyun@hit.edu.cn).}
\thanks{J. Huang is with the College of Electronics and Information Engineering, Shenzhen University, Shenzhen 518060, China (e-mail: jwhuang@szu.edu.cn).}
}

\maketitle

\begin{abstract}
The onerous media sharing task prompts resource-constrained media owners to seek help from a cloud platform, i.e., storing media contents in the cloud and letting the cloud do the sharing. There are three key security/privacy problems that need to be solved in the cloud media sharing scenario, including data privacy leakage and access control in the cloud, infringement on the owner's copyright, and infringement on the user's rights. In view of the fact that no single technique can solve the above three problems simultaneously, two cloud media sharing schemes are proposed in this paper, named FairCMS-I and FairCMS-II. By cleverly utilizing the proxy re-encryption technique and the asymmetric fingerprinting technique, FairCMS-I and FairCMS-II solve the above three problems with different privacy/efficiency trade-offs. Among them, FairCMS-I focuses more on cloud-side efficiency while FairCMS-II focuses more on the security of the media content, which provides owners with flexibility of choice. In addition, FairCMS-I and FairCMS-II also have advantages over existing cloud media sharing efforts in terms of optional IND-CPA (indistinguishability under chosen-plaintext attack) security and high cloud-side efficiency, as well as exemption from needing a trusted third party. Furthermore, FairCMS-I and FairCMS-II allow owners to reap significant local resource savings and thus can be seen as the privacy-preserving outsourcing of asymmetric fingerprinting. Finally, the feasibility and efficiency of FairCMS-I and FairCMS-II are demonstrated by experiments.
\end{abstract}

\begin{IEEEkeywords}
Cloud media sharing, asymmetric fingerprinting, proxy re-encryption, copyright protection.
\end{IEEEkeywords}

\IEEEpeerreviewmaketitle

\section{Introduction}
{\color{black}In day-to-day life, people are encountering an ever-growing volume of media big data through various social media platforms such as Facebook, Twitter, and WeChat. As a result, it has become increasingly common for media owners to share their contents with multiple users. To handle the vast number of users and media contents, the owner is typically required to deploy high-performance servers and broadband networks to cope with the onerous computing, storage, and communication demands of copyright-protected media sharing solutions \cite{hu2020cover,wang2023data}. While this may not be a challenge for media companies, it can be prohibitively expensive for individuals.}

{\color{black}An intuitive approach to reduce overhead for the owner is to store the media contents in a cloud platform and, with the help of the cloud, share the media contents to the authorized users. It is evolving into an emerging technique called cloud media sharing \cite{zhang2018you,dong2020watermarking}. In this technique,} on the one hand, the owner can make full use of the abundant software, hardware, and bandwidth resources of the cloud, thereby avoiding the expense of deploying servers and networks; on the other hand, the cloud can earn economic benefits by collecting rent from owners. 

{Despite the benefits of cloud media sharing, {\color{black}due to a variety of reasons such as curiosity, financial incentive, and reputation damage, there are still several security/privacy problems that have to be solved,} as shown below.}

\textit{Problem 1:} {Data privacy leakage and access control in the cloud. {\color{black}On the one hand, the cloud service provider could be curious about the data it encounters. On the other hand, it is a challenge to implement access control over the media content without direct control by the owner.}}

\textit{Problem 2:} The infringement on the data owner's copyright. {\color{black}Upon receiving the media content from the cloud, the authorized user could redistribute the owner’s media content arbitrarily. Clearly, it damages the owner’s copyright.} 

\textit{Problem 3:} The infringement on the user's rights. The owner may frame any user by falsely alleging that the user's water-mark has been found in an unauthorized content copy, while the fact is that this user's watermark had been abused by the owner or adversaries.
{The possibility of launching a malicious incrimination on an honest user harms the user’s rights.} 

To solve \textit{Problem 1}, cryptosystems implementing privacy-protected access control are required. {\color{black}Attribute-Based Encryption (ABE) \cite{bobba2010attribute,wei2018fractal} is such a cryptosystem, in which a user’s credentials are described by attributes. A ciphertext can be de-crypted by the user when the attributes pass the ciphertext’s access structure. Proxy Re-encryption (PRE) \cite{xu2012cl,ateniese2006improved} is another common cryptographic access control system, which allows the cloud to convert the owner’s ciphertext into a user’s ci-phertext without knowing the plaintext version.}

We then move on to discuss possible solutions to \textit{Problems 2 and 3}. 
{\color{black}A watermarking technique being able to safeguard the user's rights while maintaining the owner's copyright is called Asymmetric Fingerprinting (AFP) \cite{pfitzmann1996asymmetric, bianchi2014ttp, memon2001buyer, lei2004efficient, rial2010provably, bianchi2013secure}.} AFP mainly relies on cryptographic tools including public key cryptosystem and homomorphic encryption, in which the embedding operation is performed in the ciphertext domain so that only the user has assess to his/her own fingerprint. Since the owner cannot acquire the user's personal fingerprint, he/she has no ability to frame honest users. Nevertheless, if the owner later finds a suspicious media copy, the malicious user can still be identified and proved guilty in front of a judge.

As discussed above, AFP seems to solve \textit{Problems 2 and 3} perfectly. {\color{black}However, this is no longer the case when media contents are remotely hosted by the cloud since existing AFP schemes were designed without taking the cloud's involvement into consideration. Thus it remains to be further explored how to develop a novel AFP solution compatible with the system model of cloud media sharing.} This is a challenging task, not even to mention that we also have to simultaneously protect the privacy of the data in the cloud, i.e., solving \textit{Problem 1}. 


There are two extra challenges that need to be addressed. {\color{black}For one thing, considering that the original purpose of cloud's involvement is to help resource-constrained owners efficiently share their media contents, the owner-side overhead needs to be carefully controlled to ensure that owners can obtain sig-nificant resource savings.} For another, since the number and identity composition of target users generally cannot be determined accurately in advance, {\color{black}a practical cloud media sharing scheme should be scalable in the sense its capacity expands in real time as the number of subscriptions increases.} PRE is better than ABE if the access policy changes frequently \cite{zhang2018you, qin2016survey}. {\color{black}Therefore, we tend to use PRE for scheme construction in this paper.}

In this paper, facing these problems and challenges, we set out to solve them. First, to achieve data protection and access control, we adopt the lifted-ElGamal based PRE scheme, as discussed in \cite{yu2019file,samanthula2015secure,gao2019cloud,shafagh2017secure,derler2017homomorphic}, whose most prominent characteristic is that it satisfies the property of additive homomorphism. {\color{black}Then this homomorphism property is fully exploited to facilitate the integration with the Look-Up Table (LUT) based AFP scheme put forward by Bianchi \textit{et al.} \cite{bianchi2014ttp}.} In this way, the cloud is successfully introduced to participate in the AFP solution, and the combination of the two technologies provides an approach to solve \textit{Problems 1, 2, and 3} simultaneously. 

{\color{black}According to the above idea, we propose two cloud media sharing schemes in this paper, i.e., FairCMS-I and FairCMS-II, which solve the above three problems with different privacy/efficiency trade-offs. Among them, FairCMS-I consumes fewer cloud resources, while FairCMS-II achieves better protection for the media content.} The different encryption methods used by the two schemes for the media content stored in the cloud are responsible for this performance difference.

{\color{black}From the point of view of the owner side, both FairCMS-I and FairCMS-II can be regard as the outsourcing of the LUT-based AFP \cite{bianchi2014ttp}, i.e., the owner can reap significant savings in the local storage, communication, and computing resources.} Therefore, both schemes meet the requirement of owner-side efficiency, which is validated by the comprehensive theoretical analysis in Section~\ref{Sec:EfficiencyAnalysis} and the experimental evaluation in Section~\ref{sec:simulation}. Meanwhile, the two schemes meet the scalability requirement by performing re-encryption operations for each user. In addition, experiments are also carried out to demonstrate the excellent performance of the proposed two schemes in terms of perception quality and success tracing rate.

Finally, the major contributions of this paper are summarized as follows:
\begin{itemize}








\item {\color{black}Aiming at the situation that the existing techniques can-not fully meet the security/privacy requirements of cloud media sharing, we propose two novel schemes, namely FairCMS-I and FairCMS-II, to solve \textit{Problems 1, 2, and 3} with different privacy/efficiency trade-offs, which are also qualified in terms of owner-side efficiency and scalability. Compared to existing schemes, FairCMS-I and FairCMS-II possess the advantages outlined in Table \ref{tab:comparison}.}

\item {\color{black}By delegating the management of the media content to the cloud, FairCMS-I and FairCMS-II can also be seen as an instantiation of privacy-preserving outsourcing of AFP, thereby solving the problem caused by insufficient local resources of the owner in media sharing.}


\end{itemize}

The rest of this paper is outlined below. {\color{black}The next section reviews the related work. Section \ref{sec:problemstatement} describes the system model, threat model, and design goals. Subsequently, Section \ref{sec:fundamental} introduces the involved fundamental techniques. The two schemes are constructed in Section \ref{Scheme_Construction}.} The performance of the two schemes regarding the three problems is evaluated in Section \ref{Achieving} followed by the efficiency analysis in Section \ref{Sec:EfficiencyAnalysis}. The experimental results are reported in Section \ref{sec:simulation}. The last section concludes {the} paper. 

\section{Related Work}
\label{sec:relatedwork}
Firstly, this work inherits from the privacy-protected cloud media sharing solutions based on ABE or PRE. Wu \textit{et al.} \cite{wu2013attribute} came up with an ABE scheme with multi-message ciphertext policy, {\color{black}which was implemented for scalable media sharing and access control based on the user's attributes.} 
{\color{black}Polyakov \textit{et al.} \cite{polyakov2017fast} proposed two multihop unidirectional PRE schemes for controlled publication and subscription of cloud data,} which supports the transfer of proxy access rights. 
Liang \textit{et al.} \cite{liang2014dfa} defined and constructed a deterministic finite automata-based functional PRE scheme for public cloud data sharing without privacy leakage. 
{\color{black}After that, Li \textit{et al.} \cite{li2018multi} constructed a fine-grained and accountable access control system in the cloud, which traces suspicious access behaviors while ignores redistribution behaviors.}
Clearly, \textit{Problems 2 and 3} are not considered in these works.

Secondly, this work is related to the AFP schemes.
Rial \textit{et al.} \cite{rial2010provably} proposed a provably secure anonymous AFP scheme based on the ideal-world/real-world paradigm. 
{\color{black}Poh \textit{et al.} \cite{poh2008efficient} designed an innovative user-side AFP scheme based on the symmetric Chameleon encryption technique, which achieves significant gains in owner-side computing and communication efficiency.}
{\color{black}Afterwards, Bianchi \textit{et al.} \cite{bianchi2014ttp} proposed a LUT-based AFP scheme without involving a Trusted Third Party (TTP) based on homomorphic encryption, which also implements AFP within the user-side framework.}
Despite the fact that \textit{Problems 2 and 3} are solved in these works, \textit{Problem 1} is not mentioned.


{\color{black}Thirdly, there are also studies that deal with both privacy-protected access control and traitor tracing. Xia \textit{et al.} \cite{xia2016privacy} introduced the watermarking technique to privacy-protected content-based ciphertext image retrieval in the cloud, which can prevent the user from illegally distributing the retrieved images. However, the fairness of the traitor tracing is only realized by the involvement of a TTP in the scheme. Moreover, the encryption of image features in the scheme is not IND-CPA secure. Zheng \textit{et al.} \cite{zheng2022towards} aimed to achieve differential access control and access history hiding on the cloud while enabling fair redistribution tracing by embedding watermarks homomorphically. However, the computing overhead on the cloud side would be onerous due to the need of performing re-encryption operations and homomorphic operations on the media content. Additionally, a TTP is still required to generate and encrypt watermarks for every user. Frattolillo \textit{et al.} \cite{frattolillo2019multiparty} proposed a multi-party watermarking scheme for the cloud's environment, which is able to solve the aforementioned three problems simultaneously. However, IND-CPA security is not satisfied in the scheme due to the adoption of commutative cryptosystem. Zhang \textit{et al.} \cite{zhang2018you} combined PRE and fair watermarking to realize privacy-protected access control and combat content redistribution in the cloud, which also solves all three problems successfully.} For one thing, compared with the first scheme of Zhang \textit{et al.}, neither of our schemes requires the participation of a TTP. For another, compared with the second scheme of Zhang \textit{et al.}, which does not require a TTP, in our proposed scheme FairCMS-I, the cloud only needs to perform homomorphic operations and re-encryption operations on the encrypted LUT and fingerprint instead of the encrypted media content. {As LUTs and fingerprints are far shorter than the media content itself, FairCMS-I consumes much fewer cloud resources than that of \cite{zhang2018you} (the cloud-side overhead of the two schemes in \cite{zhang2018you} is the same). {\color{black}Furthermore, in the second scheme of Zhang \textit{et al.}, the user can escape traceability by generating two different fingerprints (we discuss this in detail in the third last paragraph of Section \ref{sec:FairCMS-I}), and both FairCMS-I and FairCMS-II solve this problem.}

\begin{table*}[t]
\centering
\caption{{\color{black}Comparison with related existing schemes}}
\label{tab:comparison}
\begin{threeparttable}
\begin{tabular}{m{1.8cm}<{\centering}m{1.6cm}<{\centering}m{1.6cm}<{\centering}m{1.6cm}<{\centering}m{1.6cm}<{\centering}m{1.6cm}<{\centering}m{2.3cm}<{\centering}m{2.3cm}<{\centering}}
\toprule
  \multirow{3}{*}{{\color{black}Schemes}} & \multirow{3}{*}{\makecell{{\color{black}Cloud’s} \\ {\color{black}involvement}}} & \multicolumn{2}{c}{\textit{{\color{black}Problem 1}}} &  \multirow{3}{*}{\makecell{{\color{black}Cloud-side} \\ {\color{black}efficiency}}} & \multirow{3}{*}{{\color{black}TTP-free}} & \multicolumn{2}{c}{\textit{{\color{black}Problems 2 and 3}}} \\
  \cline{3-4}\cline{7-8}
   & & \makecell{{\color{black}Privacy} \\ {\color{black}protection}} & \makecell{{\color{black}Access} \\ {\color{black}control}} & & & \makecell{{\color{black}Protect the} \\ {\color{black}owner's copyright}} & \makecell{{\color{black}Protect the user's} \\ {\color{black}rights}}   \\
 \midrule
 {\color{black}\cite{wu2013attribute,polyakov2017fast,liang2014dfa,li2018multi}} & {\color{black}\checkmark} & {\color{black}\checkmark} & {\color{black}\checkmark} & {\color{black}\checkmark} & {\color{black}\checkmark} & {\color{black}$\times$} & {\color{black}$\times$}\\ 
 {\color{black}\cite{poh2008efficient}} & {\color{black}$\times$} & {\color{black}$-$\tnote{*}} & {\color{black}$-$\tnote{*}} & {\color{black}$-$\tnote{*}} & {\color{black}$\times$}  & {\color{black}\checkmark} & {\color{black}\checkmark}\\ 
   {\color{black}\cite{rial2010provably,bianchi2014ttp}} & {\color{black}$\times$} & {\color{black}$-$\tnote{*}} & {\color{black}$-$\tnote{*}} & {\color{black}$-$\tnote{*}} & {\color{black}\checkmark} & {\color{black}\checkmark} & {\color{black}\checkmark} \\
   {\color{black}\cite{xia2016privacy}} & {\color{black}\checkmark} & {\color{black}$\checkmark\mkern-11mu{\smallsetminus}$\tnote{*}} & {\color{black}\checkmark} & {\color{black}\checkmark} & {\color{black}$\times$} & {\color{black}\checkmark} & {\color{black}\checkmark} \\
    {\color{black}\cite{zheng2022towards}} & {\color{black}\checkmark} & {\color{black}\checkmark} & {\color{black}\checkmark} & {\color{black}$\times$} & {\color{black}$\times$} & {\color{black}\checkmark} & {\color{black}\checkmark} \\
  {\color{black}\cite{frattolillo2019multiparty}} & {\color{black}\checkmark} & {\color{black}$\checkmark\mkern-11mu{\smallsetminus}$\tnote{*}} & {\color{black}\checkmark} & {\color{black}$\times$}  & {\color{black}\checkmark} & {\color{black}\checkmark} &  {\color{black}\checkmark}\\
 {\color{black}\cite{zhang2018you}-I\tnote{*}} & {\color{black}\checkmark} & {\color{black}\checkmark} & {\color{black}\checkmark}  & {\color{black}$\times$} & {\color{black}$\times$} & {\color{black}\checkmark} & {\color{black}\checkmark}\\
 {\color{black}\cite{zhang2018you}-II\tnote{*}} & {\color{black}\checkmark} & {\color{black}\checkmark} & {\color{black}\checkmark}  & {\color{black}$\times$} & {\color{black}\checkmark} & {\color{black}\checkmark} & {\color{black}\checkmark}\\
  {\color{black}FairCMS-I} & {\color{black}\checkmark} & {\color{black}$\checkmark\mkern-11mu{\smallsetminus}$\tnote{*}} & {\color{black}\checkmark} & {\color{black}\checkmark} & {\color{black}\checkmark} & {\color{black}\checkmark} & {\color{black}\checkmark}\\
  {\color{black}FairCMS-II} & {\color{black}\checkmark} & {\color{black}\checkmark} & {\color{black}\checkmark} & {\color{black}$\times$} & {\color{black}\checkmark} & {\color{black}\checkmark} & {\color{black}\checkmark}\\
 \bottomrule
\end{tabular}
\begin{tablenotes}
\footnotesize
{\color{black}\item[*] `$-$' indicates that the property is not scored because the involvement of cloud is not considered. $\checkmark\mkern-11mu{\smallsetminus}$ means that the privacy of cloud media is protected, but that protection is not IND-CPA secure. \cite{zhang2018you}-I and \cite{zhang2018you}-II represent the first scheme and the second scheme in \cite{zhang2018you}, respectively. Compared with \cite{zhang2018you}-II, the main advantage of FairCMS-II is that it solves the problem that users can escape traceability by generating two different fingerprints, as discussed in the third last paragraph of Section \ref{sec:FairCMS-I}.}
\end{tablenotes}
\end{threeparttable}
\vspace{-2pt}
\end{table*}

{\color{black}Finally, the comparison between the two proposed schemes and the existing relevant schemes is summarized in Table \ref{tab:comparison}. As can be seen therein, the two proposed schemes FairCMS-I and FairCMS-II have advantages over the existing works. In addition, the two proposed schemes offer owners the flexibility to choose. If the security requirements for the media content are not excessively rigorous and the size of the media content is small (e.g., images with a moderate pixel count), the owner can choose FairCMS-I to minimize the cost of renting cloud resources; otherwise, the owner can choose FairCMS-II. There is no fixed security requirement or content size threshold to guide the selection between these two options. Instead, it is up to the owner to make a decision based on the objective application scenario and his/her subjective considerations. In Section \ref{sec:simulation}, we conduct a comparative experiment on the cloud-side efficiency of FairCMS-I and FairCMS-II to provide a quantitative reference for the owner's decision-making.}

\section{Problem Statement}
\label{sec:problemstatement}
\subsection{\color{black}System Model}}
As illustrated in Fig.~\ref{fig:sysmodel}, our system model mainly consists of four entities: the owner, the cloud, users, and the judge. Their roles are elaborated as follows.

\begin{figure}[ht]
\centering
\includegraphics[scale=0.06]{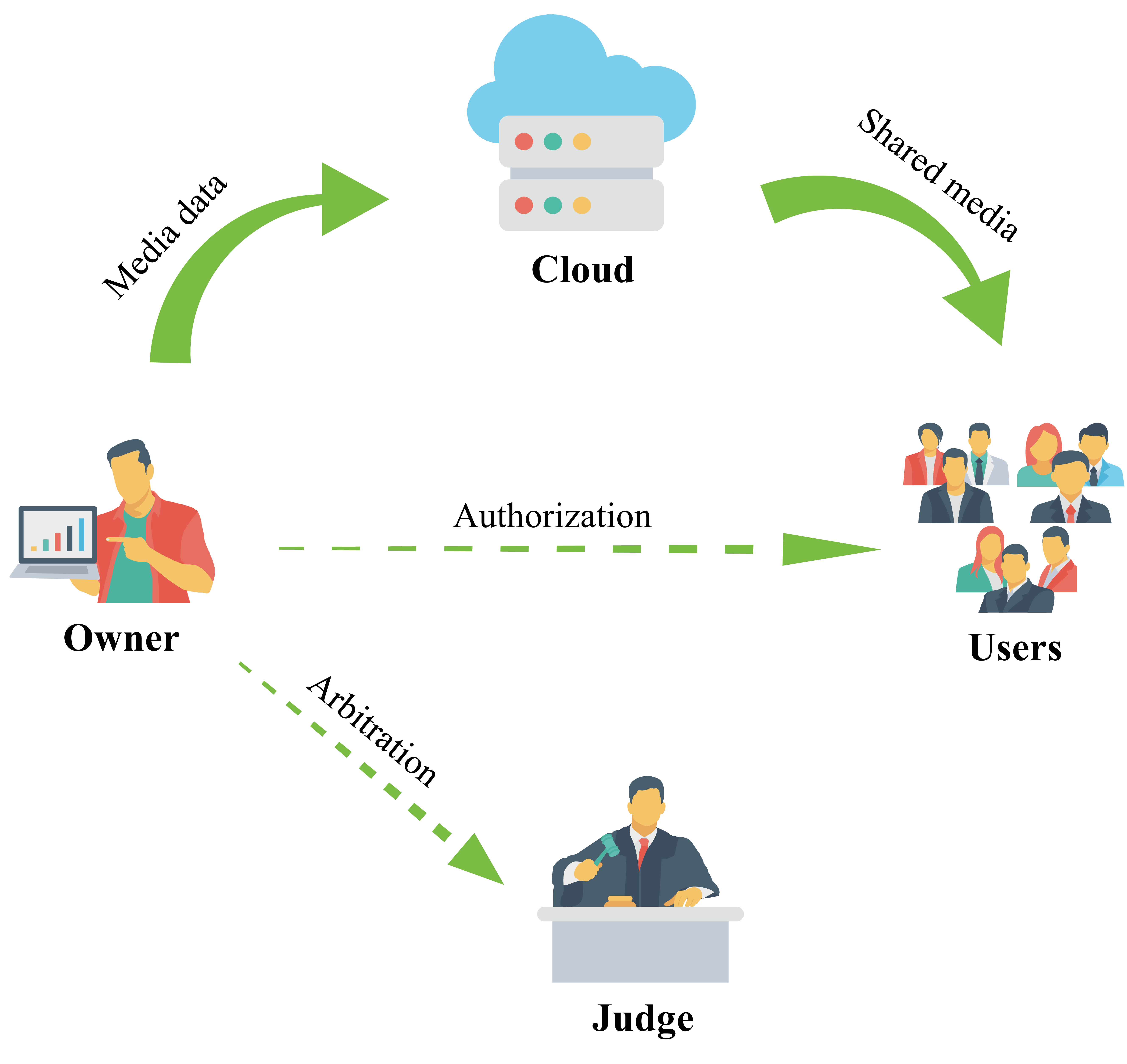}
\caption{{System Model.}}
\label{fig:sysmodel}
\vspace{-2pt}
\end{figure}

\begin{itemize}
\item \textbf{Owner}. The owner, who might be a media producer or trader, {\color{black}plans to store the media content in the cloud and let the cloud do the content distribution. 
First, the owner requires that the cloud not be able to obtain the plaintext about the media content and the LUTs, and that access to the media content is controlled by his/her authorization.} 
Second, the owner asks for significant overhead savings from cloud media sharing. Third, the owner demands traitor tracing of users who violate copyright. 
{\color{black}If a suspicious redistributed copy of the media content is found,} the owner can file an arbitration claim with the judge.

\item \textbf{Cloud}. The cloud has a wealth of hardware, software, and bandwidth resources, part of which are lent to the owner through a pay-per-use way to help the owner carry out media sharing. 

\item \textbf{Users}. {\color{black}Users want to access the owner's media content. To this end, users request authorization from the owner, for example by paying for purchases. If successful, users can get the desired shared media content from the cloud. Users require that the plaintext of their fingerprints not be accessed by the owner or the cloud, to prevent malicious framing by the owner.}

\item \textbf{Judge}. {\color{black}The judge is a trusted entity who is only responsible for arbitration in the case of illegal redistribution, as in existing traitor tracing systems \cite{bianchi2014ttp, memon2001buyer, lei2004efficient, rial2010provably, bianchi2013secure,zhang2018you}. After receiving the owner's request for arbitration, the judge makes a fair judgment based on the evidence provided by the owner. Although only the encrypted version of the user's watermark is disclosed, the encrypted watermark can be converted into a ciphertext that can be decrypted by the judge based on PRE (for details, please see Figs. \ref{fig:3} and \ref{fig:4}), thus enabling traitor tracing.} Once the judge detects a copyright infringement, the unfaithful user will be prosecuted in accordance with the law.
\end{itemize}

\subsection{{\color{black}Threat Model}}
{\color{black}The threats considered in this paper come from three entities: users, the owner, and the cloud. First, users are assumed to be malicious, who could illegally redistribute the owner's media content with the hope that this behavior will not be detected. Second, the owner is also assumed to be malicious, who may try to obtain the users' fingerprints and maliciously embed the fingerprints into any media content to frame honest users for copyright infringement. Third, the cloud is assumed to be honest-but-curious, which is the same as other privacy-preserving cloud media sharing schemes based on ABE or PRE \cite{xu2012cl,wu2013attribute,liang2014dfa,seo2014efficient,qin2016survey}. Although the honest-but-curious cloud faithfully performs its assigned duties, he/she could try to steal the plaintext about the owner's media content. Moreover, the cloud is also curious about other information it encounters, including the users' fingerprints and the LUTs. Finally, as in \cite{frattolillo2019multiparty}, we assume that there may exist collusion among individual users and collusion between the owner and the cloud, while there is no collusion between users and the cloud.}


\subsection{Design Goals}
\label{subsec:designgoals}
Under the above system model and threat model, we summarize the design goals as follows.

\begin{itemize}
\item \textbf{{Implement privacy-preserving access control}}. {On the one hand, the cloud should be prevented from obtaining the private plaintext of the data it encounters, {\color{black}including the owner's media content, the users' fingerprints, and the LUTs. On the other hand, only users authorized by the owner can access the media content.}}

\item \textbf{Protect the owner's copyright}. {\color{black}We need to embed the user's fingerprint in the owner's media content to enable traitor tracing. As long as an unfaithful user makes an unauthorized redistribution, he/she can be detected by the embedded fingerprint in the media content.}

\item \textbf{Protect the user’s {right}}. {\color{black}We need to prevent the owner from framing an honest user for copyright infringement by embedding this user's fingerprint in any media content.}

\item \textbf{{\color{black}Ensure efficiency gains and scalability}}. {\color{black}For one thing, we need to carefully control the owner-side overhead to ensure that the owner can gain significant local resource savings from cloud media sharing. For another, we need to ensure that the two proposed schemes are scalable to handle real-time requests from users.}

\item \textbf{{\color{black}TTP-free}}. {\color{black}The two proposed schemes should not require any TTP to participate in the media sharing process. The TTP mentioned here does not cover the judge, who is only responsible for handing down sentences in cases of suspected redistribution and is not involved in the media sharing process. The judge is an indispensable participant. In contrast, the existing schemes \cite{xia2016privacy,zheng2022towards,zhang2018you} that are not TTP-free all require a TTP to remain online during media sharing to generate watermarks for users.}
\end{itemize}

\section{Fundamental Techniques}
\label{sec:fundamental}

{\subsection{LUT-Based Secure Embedding}
For better understanding, we start with the introduction of a LUT-based watermarking scheme \cite{celik2008lookup, celik2007secure}, which is the basis of the AFP scheme proposed by Bianchi \textit{et al.} \cite{bianchi2014ttp}.

\subsubsection{Key Generation}
First, the owner creates a $T \times 1$ encryption LUT (E-LUT) $\mathbf{E}$ with its entry ${\mathbf{E}}( t )~(0\leq t \leq T-1)$ following the Gaussian distribution $N( {0,{\sigma _E}} )$. Then, for the $k$-th user, his/her personalized decryption LUT (D-LUT) $\mathbf{D}_k$ is obtained by modifying the E-LUT as
\begin{equation*}
\mathbf{D}_k =  - \mathbf{E} + \mathbf{W}_k,
\end{equation*}
where ${\mathbf{W}_k}$ is the $k$-th user's personalized watermarking LUT (W-LUT) that is generated by
\begin{equation*}
{\mathbf{W}}_k = \mathbf{G}{\mathbf{w}}_k, 
\end{equation*}
where ${\mathbf{G}}$ is a $T \times L$ encoding matrix \cite{bianchi2014ttp,Marshall2003Coding,wang2003complex} and the vector ${\mathbf{w}}_k = [w_{k,l}]_{l=0}^{L-1}$ is {\color{black}computed by ${w_{k,l}} = {\sigma _W}( {2{b_{k,l}} - 1} )$. 
Here, ${\sigma _W}$ is a public parameter utilized to set the standard deviation of ${{\mathbf{W}}_k}$} and ${\mathbf{b}_k = [b_{k,l}]_{l=0}^{L-1}}$ is the $k$-th user's $L$-bit uniform binary fingerprint.}

\subsubsection{Encryption}
Based on the E-LUT $\mathbf{E}$, {\color{black}the owner encrypts the media content using the single-value alteration method.} 
{Assume that the media content can be represented as a vector $\mathbf{m}$ of length $M$, the encryption method used here can be expressed as}
\begin{equation}
\mathbf{c} = \mathbf{m} + \mathbf{B}^m \mathbf{E},
\label{Eq:encryption}
\end{equation}
where $\mathbf{B}^m$ is a $M \times T$ binary matrix defined as
\begin{equation*} 
\mathbf{B}^m( i,j ) = 
\left\{ 
\begin{aligned} 
&1,~~~j = {t_{ih}}, \\
&0,~~~\text{otherwise},
\end{aligned}
\right.
\end{equation*}
where $t_{ih}~(0 \leq i \leq M-1, 0 \leq h \leq S-1)$ is a set of $M \times S$ values in the range $[0,T-1]$. The generation of $t_{ih}$ is controlled by an index generator under a session key $SK_m$. 


\subsubsection{Joint Decryption and Fingerprinting}
Upon receiving the $k$-th user’s request, the owner sends $\mathbf{c}$ to the $k$-th user who performs joint decryption and fingerprinting with the D-LUT ${\mathbf{D}}_k$ to obtain his/her personalized fingerprinted media content $\mathbf{m}^k$ by performing
\begin{equation} 
\mathbf{m}^k = \mathbf{c} +  \mathbf{B}^m \mathbf{D}_k.
\label{Eq:decryption}
\end{equation}
Note that Eq.~(\ref{Eq:decryption}) can be calculated as
\begin{IEEEeqnarray}{rCL}
\mathbf{m}^k &=& {\mathbf{c}} + {{\mathbf{B}}^m}{{\mathbf{D}}_k} \nonumber\\
&=& {\mathbf{m}} + {{\mathbf{B}}^m}{\mathbf{E}}{\rm{ + }}{{\mathbf{B}}^m}{{\mathbf{D}}_k} \nonumber\\
&=& {\mathbf{m}} + {\sigma _W}{\bar{\mathbf G}}{({2{{\mathbf{b}}_k}}-1)},
\label{Eq:decandfinger}
\end{IEEEeqnarray}
where ${\bar{\mathbf G}} = {{\mathbf{B}}^m}{\mathbf{G}}$. It is clear from Eq.~(\ref{Eq:decandfinger}) that the fingerprint $\mathbf{b}_k$ has been successfully embedded into the original media content $\mathbf{m}$ under the modulation of the secret matrix $\bar{\mathbf G}$. 

\subsubsection{Fingerprint Detection and Traitor Tracing}
Once a copyright dispute occurs between the owner and the user, {\color{black}they delegate a judge that is credible for both parties to make a fair arbitration.} Due to the possible noise effect during data transmission, {\color{black}the received suspicious media content copy is assumed to be contaminated by the an additive noise $\mathbf{n}$,} i.e.,
\begin{equation*}
\left. 
\begin{aligned}
\tilde{\mathbf m}^k &= {\mathbf{m}} + {\bar{\mathbf G}}{\tilde{{\mathbf{w}}}}_k + {\mathbf{n}}\\
&={\mathbf{m}} + {\sigma _W}{\bar{\mathbf G}}{({2{\tilde{\mathbf{b}}_k}}-1)} + {\mathbf{n}},
\end{aligned}
\right.
\end{equation*}
where ${\tilde{\mathbf{b}}_k}$ is the suspicious user's fingerprint.

In order to detect the fingerprint ${\tilde{\mathbf{b}}_k}$, the judge can leverage the suboptimal decoders such as the Matched Filter decoder and the Pseudo-Inverse decoder \cite{bianchi2014ttp}, which are respectively formulated as
\begin{IEEEeqnarray}{rCL}
 \tilde{\mathbf b}_k = {\mathop{\rm sgn}} \left\{ {{{{\bar{\mathbf G}}}^T}( {{{{\tilde{\mathbf m}}}^k} - {\mathbf{m}}} )} \right\},
 \label{Eq:mfdecoder}
\end{IEEEeqnarray}
and 
\begin{IEEEeqnarray}{rCL}
\tilde{\mathbf b}_k = {\mathop{\rm sgn}} \left\{ {( {{{{\bar{\mathbf G}}}^T}{\bar{\mathbf G}}} )^{-1}{{{\bar{\mathbf G}}}^T}( {{{{\tilde{\mathbf m}}}^k} - {\mathbf{m}}} )} \right\},
 \label{Eq:pidecoder}
\end{IEEEeqnarray}
where 
\begin{equation*}
{\mathop{\rm sgn}} \left\{ \zeta  \right\} = \left\{ \begin{aligned} 1,~~~ \zeta  > 0,\\
0,~~~ \zeta  \le 0.
\end{aligned} \right. 
\end{equation*}
If ${{\tilde{\mathbf{b}}}_k}={{\mathbf{b}}_k}$ is established within the allowable error range, the judge rule that the $k$-th user has illegally redistributed the owner's media; otherwise, the $k$-th user is proved innocent.


{\subsection{AFP Based on User-Side Embedding}
\label{subsec:afp}
Given the above arithmetics of the LUT-based user-side embedding in the plaintext domain, the additive homomorphism is used to implement the secure distribution of D-LUTs in the ciphertext domain \cite{bianchi2014ttp}.}


\subsubsection{Additive {Homomorphism}}
Suppose that $E_{PK}(\cdot)$ is the additive homomorphic encryption under the public key $PK$. By 
saying additive {homomorphism}, $E_{PK}(\cdot)$ satisfies the following equations:
\begin{equation} 
 \begin{aligned}
&E_{PK}(m_1+m_2)=E_{PK}(m_1) \cdot E_{PK}(m_2),\\
&E_{PK}(a \cdot m)=E_{PK}(m)^a,
\end{aligned}
\label{Eq:addhom}
\end{equation}
{\color{black}where $m_1$, $m_2$, and $m$ are plaintexts, and $a$ is an integer.}

\subsubsection{Secure Distribution of D-LUTs}
First, the $k$-th user generates his/her fingerprint $\mathbf{b}_k$ locally and encrypts it with his/her public key $PK_{U_k}$, i.e., $E_{PK_{U_k}}( \mathbf{b}_k ) = [ E_{PK_{U_k}}( b_{k,0} ),$ $ E_{PK_{U_k}}( b_{k,1} ),\cdots, E_{PK_{U_k}}( b_{k,L-1} ) ]$. 
Then, the $k$-th user sends $E_{PK_{U_k}}( {{\mathbf{b}}_k} )$ to the owner. After receiving $E_{PK_{U_k}}( {{\mathbf{b}}_k} )$, the owner calculates $\mathbf{w}_k$ in the ciphertext domain by
\begin{IEEEeqnarray}{rCL}
E_{PK_{U_k}}( w_{k,l} ) &=& E_{PK_{U_k}}( 2b_{k,l} - 1 )^{\sigma_W} \nonumber \\
    &=& E_{PK_{U_k}}{( b_{k,l} )^{2 \sigma _W}} \cdot E_{PK_{U_k}}{( 1 )^{-{\sigma _W}}}, 
\label{Eq:enccalwkl}
\end{IEEEeqnarray}
and further calculates
\begin{IEEEeqnarray}{rCL}
&& E_{PK_{U_k}}( \mathbf{D}_k ( t ) )   \nonumber \\
 && =  E_{PK_{U_k}}{( {\mathbf{E}( t )} )^{ - 1}} \cdot \prod\limits_{l = 0}^{L-1} {E_{PK_{U_k}}{{( {w_{k,l}} )}^{{\mathbf{G}}( {t,l} )}}} ,
\label{Eq:enccaldkl}
\end{IEEEeqnarray}
where $0\le t \le T-1$. $E_{PK_{U_k}}( \mathbf{D}_k(t) )$ is then sent to the $k$-th user, who can acquire his/her personalized D-LUT through decryption with his/her private key $SK_{U_k}$.

{\color{black}Through the above steps, the owner and the $k$-th user are prevented from knowing the plaintext of $\mathbf{b}_k$ and the plaintext of $\mathbf{G}$ (and $\bar{\mathbf{G}}$) respectively,} i.e., neither of them know the value of the embedded watermark $\bar{\mathbf{G}}\mathbf{w}_k$, thus achieving simultaneous protection of the owner's copyright and the users' rights.

\subsection{Additive Homomorphic PRE}
\label{subsec:pre}

The lifted-ElGamal based PRE scheme \cite{yu2019file,samanthula2015secure,gao2019cloud,shafagh2017secure,derler2017homomorphic} is summarized as follows.

Let $G_1$, $G_2$ be two multiplicative cyclic groups of prime order $q$ with a bilinear map $e:~{G_1} \times {G_1} \to {G_2}$, and $g$ be a random generator of $G_1$. The mapping $e$ has three properties. 1) Bilinearity: for any $a,b\ {\in}\ \mathbb{Z}_q$, $e(g^a,g^b)=e(g,g)^{ab}$. 2) Non-degeneracy: $e(g,g)\ {\neq}\ 1$. 3) Computability: $e$ can be efficiently computed.

With bilinear pairing, the lifted-ElGamal based PRE technique is comprised of the following algorithms by using the public parameters $G_1,G_2,q,e,g$ and $Z=e(g,g)\ {\in}\ G_2$.
 
\begin{itemize}
\item \textbf{Key Generation ($KG$)}. A user $A$’s key pair is of the form $PK_A=(Z^{a_1},g^{a_2})$ and $SK_A=(a_1,a_2)$, and a user $B$'s key pair is of the form $PK_B=(Z^{b_1},g^{b_2})$ and $SK_B=(b_1,b_2)$.

\item \textbf{Re-encryption Key Generation ($RG$)}. A user $A$ {authorizes} $B$ access to $A$'s data by publishing the re-encryption key $rk_{A{\rightarrow}B}=(g^{b_2})^{a_1}=g^{a_1b_2}\ {\in}\ G_1$, computed from $A$’s private key and $B$’s public key\footnote{{\color{black}One single user $A$ can also generate a re-encryption key for himself to change the second-level ciphertext into the first-level,} we denote this key as $rk_{A^2{\rightarrow}A^1}=(g^{a_2})^{a_1}=g^{a_1a_2}$.}.

\item \textbf{First-Level Encryption ($E^1$)}. To encrypt a message $m \in Z_q$ under $PK_A$ in such a way that it can only be decrypted by the holder of $SK_A$, the scheme computes $E^1_{PK_A}(m)=(Z^{a_1r},Z^mZ^r)$, where $r$ is a random num-ber from $Z_q$.

\item \textbf{Second-Level Encryption ($E^2$)}. To encrypt a message $m \in Z_q$ under $PK_A$ in such a way that it can be decrypted by $A$ and his/her delegatees, the scheme outputs $E^2_{PK_A}(m)=(g^r,Z^mZ^{a_1r})$.

\item \textbf{Re-encryption ($R$)}. Anyone can change a second-level ciphertext for $A$ into a first-level ciphertext for $B$ with $rk_{A{\rightarrow}B}=g^{a_1b_2}$. From $E^2_{PK_A}(m)=(g^r,Z^mZ^{a_1r})$, the scheme computes $e(g^r,g^{a_1b_2})=Z^{b_2a_1r}$ and publishes $E^1_{PK_B}(m)=(Z^{b_2a_1r},Z^mZ^{a_1r})=(Z^{b_2{r'}},Z^mZ^{r'})$.

\item \textbf{Decryption ($D^1,D^2$)}. To decrypt a first-level ciphertext $E^1_{PK_A}(m)=(Z^{a_ir},Z^mZ^r)=(\alpha,\beta)$ with private key $a_i\ {\in}\ SK_A$, the scheme computes $m'=Z^m=\frac{\beta}{\alpha^{1/a_i}}$ for $i\ {\in}\ \{1,2\}$. To decrypt a second-level ciphertext $E^2_{PK_A}(m)=(g^r,Z^mZ^{a_1r})=(\alpha,\beta)$ with secret key $a_1\ {\in}\ SK_A$, the scheme outputs $m'=Z^m=\frac{\beta}{e(\alpha,g)^{a_1}}$.
\end{itemize}

{\color{black}To get the plaintext $m$, the above mentioned PRE scheme requires computing the discrete logarithm of $m'$ in the base $Z$,} i.e., $m=\log_Zm'$. Although this sounds restrictive, there is no limitation to the use of the lifted-ElGamal based PRE scheme as long as the plaintext space is small, as also pointed out by many other literature works \cite{samanthula2015secure,gao2019cloud,shafagh2017secure,derler2017homomorphic}. {In our experiment, the lookup table method \cite{CHEN2022} is used to solve the discrete logarithm.} It is worth emphasizing that the private key is still secure because of its huge value space. {For example, in the experiment we quantify the plaintext to $20$ bits, whereas the private key is typically $1,024$ bits or longer.} Note that the value space doubles for each additional bit.

{\color{black}The most important feature of this PRE scheme is that both the first-level ciphertext and the second-level ciphertext satisfy the property of additive homomorphism, i.e.,}
\begin{equation*}
\begin{aligned}
 E^1_{PK_A}(m_1) \cdot E^1_{PK_A}(m_2)&=(Z^{a_ir_1}{\cdot}Z^{a_ir_2},Z^{m_1}Z^{r_1}{\cdot}Z^{m_2}Z^{r_2})\\
 &=(Z^{a_i(r_1+r_2)},Z^{m_1+m_2}Z^{r_1+r_2})\\
 &=E^1_{PK_A}(m_1+m_2),\\
  E^2_{PK_A}(m_1) \cdot E^2_{PK_A}(m_2)&=(g^{r_1}{\cdot}g^{r_2},Z^{m_1}Z^{a_ir_1}{\cdot}Z^{m_2}Z^{a_ir_2})\\
 &=(g^{r_1+r_2},Z^{m_1+m_2}Z^{a_i(r_1+r_2)})\\
 &=E^2_{PK_A}(m_1+m_2),
\end{aligned}
\end{equation*}
where $r_1$ and $r_2$ are random numbers in $Z_q$, and $i\ {\in}\ \{1,2\}$. In addition, we emphasize that the second line in Eq.~(\ref{Eq:addhom}) is also satisfied, and the verification is similar to the above.

\section{Scheme Construction}\label{Scheme_Construction}
In this section, we bring forward two cloud media sharing schemes, namely FairCMS-I and FairCMS-II. {\color{black}FairCMS-I essentially delegates the re-encryption management of LUTs to the cloud, thus significantly reducing the overhead of the owner side. Nevertheless, FairCMS-I cannot achieve IND-CPA security for the media content. Therefore, we further propose a more secure scheme FairCMS-II, which delegates the re-encryption management of both media content and LUTs to the cloud.}

\subsection{Scheme of FairCMS-I}
\label{sec:FairCMS-I}
{\color{black}In the user-side embedding AFP, since the encrypted media content shared with different users is the same, the encryption of the media content is only executed once. In contrast, due to the personalization of D-LUTs, once a new user initiates a request, the owner must interact with this user to securely distribute the D-LUT under the support of homomorphic en-cryption.} This cost scales linearly with the number of users. It is clear that the biggest source of overhead for the owner is the management and distribution of LUTs. {\color{black}Therefore, the focus of implementing resource-saving cloud media sharing is to find ways to transfer this overhead to the cloud.} Based on this {observation}, the first scheme is as follows.

\noindent \textit{Key, Fingerprint, E-LUT{,} and Encoding Matrix Generation:}

Based on the lifted-ElGamal based PRE scheme introduced in Section~\ref{subsec:pre}, the $k$-th user $U_k$ generates his/her public key $PK_{U_k}=(Z^{a_1},g^{a_2})$ and private key $SK_{U_k}=(a_1,a_2)$. {\color{black}In addition, the owner $O$ generates his/her own public key $PK_O=(Z^{b_1},g^{b_2})$ and private key $SK_O=(b_1,b_2)$,} and the judge $J$ also generates the key pair as $PK_J=(Z^{c_1},g^{c_2})$ and $SK_J=(c_1,c_2)$. Moreover, the $k$-th user generates his/her private fingerprint $\mathbf{b}_k$ locally, and the owner generates a E-LUT $\mathbf{E}$ and a encoding matrix $\mathbf{G}$ locally.

\noindent \textit{Calculation and Distribution of D-LUTs:}

{\color{black}The $k$-th user encrypts his/her private fingerprint $\mathbf{b}_k$ into a second-level ciphertext $E_{PK_{U_k}}^2 ( \mathbf{b}_k )$ using his/her public key $Z^{a_1}$, namely $E_{PK_{U_k}}^2 (\mathbf{b}_k)=[(g^{r_0},Z^{b_{k,0}}Z^{a_1r_0}),\cdots,(g^{r_{L-1}},$ $Z^{b_{k,L-1}}Z^{a_1r_{L-1}})]$, and then sends $E_{P{K_{{U_k}}}}^{2}( {{{\mathbf{b}}_k}} )$ to the cloud.} 
In addition, the $k$-th user also generates a re-encryption key $rk_{{U_k^2} \to {U_k^1}} =g^{a_1a_2}$ with his/her private key $SK_{U_k}=a_1$ and his/her public key $PK_{U_k}=g^{a_2}$, {\color{black}and generates another re-encryption key $r{k_{{U_k} \to J}}=g^{a_1c_2}$ with his/her private key $SK_{U_k}=a_1$ and the judge's public key $PK_{J}=g^{c_2}$, then sends the two re-encryption keys to the cloud.}

The owner encrypts the E-LUT $\mathbf{E}$ into a second-level ciphertext $E^2_{PK_O}( \mathbf{E}(t) )~(0 \leq t \leq T-1)$ using his/her public key $PK_O=Z^{b_1}$, i.e., $E^2_{PK_O} ( \mathbf{E}(t) )=(g^r, Z^{\mathbf{E}(t)} Z^{b_1r})$, and sends $E^2_{PK_O}( \mathbf{E}(t) )$ together with $\mathbf{G}$ to the cloud, who then stores $E_{P{K_{O}}}^{2}( {{\mathbf{E}}(t)} )$ and $\mathbf{G}$. 
In addition, the owner $O$ generates a re-encryption key $rk_{O \to {U_k}}$ under his/her private key $S{K_O}{\rm{ = }}{b_1}$ and the $k$-th user’s public key $P{K_{{U_k}}}{\rm{ = }}{g^{a_2}}$, i.e., $r{k_{O \to {U_k}}} = {g^{b_1a_2}}$, and then sends $r{k_{O \to {U_k}}}$ to the cloud.

The cloud uses the re-encryption key $rk_{O \to {U_k}}$ to change the received encrypted E-LUT under the owner's key, i.e., $E^2_{PK_O}( {{\mathbf{E}}(t)} )$, into $E^1_{PK_{U_k}}( {{\mathbf{E}}(t)} )$, which is the first-level ciphertext of the $k$-th user. 
{\color{black}Furthermore, the cloud changes $E_{PK_{U_k}}^2( \mathbf{b}_k )$ into $E_{PK_{U_k}}^1( \mathbf{b}_k )$ using $rk_{{U^2_k} \to {U^1_k}}$ and changes $E_{PK_{U_k}}^2( \mathbf{b}_k )$ into $E^1_{PK_J}( \mathbf{b}_k )$ using $rk_{{U_k} \to J}$, and then stores $E_{PK_J}^1( \mathbf{b}_k )$ into a set $\mathcal{F}$.}
{\color{black}Subsequently, the cloud calculates $E_{PK_{U_k}}^1( \mathbf{D}_k(t) )$ with $E_{PK_{U_k}}^1( \mathbf{E}(t) )$, $E^1_{PK_{U_k}}( \mathbf{b}_k )$, and $\mathbf{G}$ by Eqs.~(\ref{Eq:enccalwkl}) and (\ref{Eq:enccaldkl}).} We emphasize that this process can be performed successfully because the first level ciphertext of the lifted-ElGamal based PRE scheme satisfies the property of additive homomorphism. 
Finally, the cloud sends $E^1_{PK_{U_k}}( \mathbf{D}_k(t) )$ back to the $k$-th user, who acquires his/her personalized D-LUT $\mathbf{D}_k$ through decryption with his/her private key $SK_{U_k} = a_2$.

\noindent \textit{{\color{black}Cloud-Based Media Sharing:}}

{\color{black}The owner encrypts the media content collection $\{ \mathbf{m} \}$ to be shared using the single-value alteration method and sends the resulting encrypted media content collection $\{ \mathbf{c} \}$ to the cloud, who then stores the encrypted media content collection.}
In addition, the owner sends the binary matrix collection $\{ \mathbf{B}^m \}$ used in the encryption process to the cloud, {\color{black}which is generated locally by the owner based on the session keys and is utilized for encrypting each media content.}

The cloud calculates the secret matrix collection $\{ \bar{\mathbf{G}} \}$ with the binary matrix collection $\{ \mathbf{B}^m \}$ and the encoding matrix $\mathbf{G}$, and then stores the secret matrix collection and the binary matrix collection. Suppose the $k$-th user expects to access one of the owner's media content ${\mathbf{m}}$. {\color{black}The cloud then extracts the encryption result $\mathbf{c}$ of the media content ${\mathbf{m}}$ from the encrypted media content collection,} i.e., ${\mathbf{c}} = {\mathbf{m}} + {{\mathbf{B}}^m}{\mathbf{E}}$, and sends $\mathbf{c}$ to this user. 
Meanwhile, the corresponding binary matrix $\mathbf{B}^m$ is also extracted from the binary matrix collection by the cloud and sends to the $k$-th user. 
{\color{black}Finally, the $k$-th user decrypts $\mathbf{c}$ with his/her personalized D-LUT $\mathbf{D}_k$ and the binary matrix ${\mathbf{B}}^m$ based on Eq.~(\ref{Eq:decandfinger}) to get the fingerprint embedded media content ${\mathbf{m}}^k$.}

\noindent \textit{Arbitration and Traitor Tracing:}

{\color{black}Upon the detection of a suspicious media content copy $\tilde{\mathbf{m}}^k$, the owner resorts to the judge for violation identification.} To this end, the proofs that the owner needs to provide the judge includes the original media content $\mathbf{m}$ with no fingerprints embedded, the corresponding secret matrix $\bar{\mathbf{G}}$, and the set $\mathcal{F}$ that holds all the users' fingerprints. Among them, $\bar{\mathbf{G}}$ and $\mathcal{F}$ is available for download by the owner from the cloud. {\color{black}It is worth mentioning that the fingerprints stored in set $\mathcal{F}$ are encrypted by the judge's public key $PK_J$, so the user's fingerprint will not be leaked to the owner and the cloud, but the plaintext of the fingerprints can be decrypted by the judge with his/her own private key $SK_J=c_2$. With these materials, the judge can make a fair judgment by Eq.~(\ref{Eq:mfdecoder}) or (\ref{Eq:pidecoder}).} 

\noindent \textit{{\color{black}Summary and Discussion:}}

{\color{black}The whole FairCMS-I scheme is summarized as follows.
First, suppose an owner rents the cloud's resources for media sharing, the owner and the cloud execute \textbf{Part 1} as shown in Fig. \ref{fig:2}. Then, suppose the $k$-th user makes a request indicating that he/she wants to access one of the owner’s media content $\mathbf{m}$, the involved entities execute \textbf{Part 2} after the $k$-th user is authorized by the owner as shown in Fig. \ref{fig:3}. Once a suspicious media content copy ${\tilde{\mathbf{m}}}^k$ is detected, the owner resorts to the judge for violation arbitration, i.e., the owner and the judge jointly execute \textbf{Part 3} as shown in Fig. \ref{fig:4}.}

\begin{figure}[!t]
  \centering
  
  \fbox{
    \begin{minipage} [!t]{0.46\textwidth}
      {\textbf{Input:} Security parameter $\lambda$.}
      
      {\textbf{Output:} The cloud stores $\{ \bar{\mathbf{G}} \}$, $\mathbf{G}$, $E^2_{PK_O} (\mathbf{E}(t))$, $\{\mathbf{B}^m \}$, and $\{ \mathbf{c} \}$.}
      
      {\textbf{Procedure:}}
      
      {\underline{Owner:} 
      \begin{itemize}
        \item Generate $( \mathbf{E}, \mathbf{G},PK_O,SK_O,\{SK_m\} )$,
        \item $E^2_{PK_O}( \mathbf{E}( t )) \leftarrow ( \mathbf{E}(t),PK_O )$,
        \item $\{ t_{ih} \} \leftarrow \{ SK_m \}$,
        \item $\{ \mathbf{B}^m \} \leftarrow \{ t_{ih} \}$,
        \item $\{ \mathbf{c} \} \leftarrow ( \{ \mathbf{m} \},\{ \mathbf{B}^m \}, \mathbf{E} )$,
        \item Send $\mathbf{G}$, $E^2_{PK_O}( \mathbf{E} (t))$, $\{ \mathbf{B}^m \}$, and $\{ \mathbf{c} \}$ to the cloud.
      \end{itemize}
      
      \underline{Cloud:}
      \begin{itemize}
        \item $\{ \bar{\mathbf{G}} \} \leftarrow ( \{ \mathbf{B}^m \},\mathbf{G} )$,
        \item Store $\{ \bar{\mathbf{G}} \}$, $\mathbf{G}$, $E^2_{PK_O} (\mathbf{E}(t))$, $\{\mathbf{B}^m \}$, and $\{ \mathbf{c} \}$.
      \end{itemize}} 
    \end{minipage}
  }
  \caption{{FairCMS-I: Media storage (\textbf{Part 1}).}}
  \label{fig:2}
     \vspace{-2pt}
\end{figure}

\begin{figure}[!t]
  \centering
  
  \fbox{
    \begin{minipage} [!t]{0.46\textwidth}
      {\textbf{Input:} Security parameter $\lambda$, $\mathbf{G}$, $E^2_{PK_O} (\mathbf{E}(t))$, $\{\mathbf{B}^m \}$, and $\{ \mathbf{c} \}$.}
      
      {\textbf{Output:} The $k$-th user obtains the watermarked media content $\mathbf{m}^k$.}
      
      {\textbf{Procedure:}}
      
      {\underline{$k$-th User:} 
      \begin{itemize}
        \item Generate $( \mathbf{b}_k,PK_{U_k},SK_{U_k} )$,
        \item $E^2_{PK_{U_k}} ( \mathbf{b}_k ) \leftarrow ( \mathbf{b}_k,PK_{U_k} )$,
        \item $rk_{U^2_k \to U^1_k} \leftarrow ( SK_{U_k},PK_{U_k} )$,
        \item $rk_{{U_k} \to J} \leftarrow ( SK_{U_k},PK_J )$,
        \item Send $E^2_{PK_{U_k}}( \mathbf{b}_k )$, $rk_{U^2_k \to U^1_k}$, and $rk_{U_k \to J}$ to the cloud.
      \end{itemize}
      
      \underline{Owner:}
      \begin{itemize}
        \item $rk_{O \to U_k} \leftarrow ( SK_{O},PK_{U_k} )$,
        \item Send $rk_{O \to U_k}$ to the cloud.
      \end{itemize} 
      
      \underline{Cloud:}
      \begin{itemize}
        \item $E^1_{PK_{U_k}} (\mathbf{E}(t)) \leftarrow ( E^2_{PK_O}(\mathbf{E}(t)),rk_{O \to U_k} )$,
        \item $E^1_{PK_{U_k}}( \mathbf{b}_k ) \leftarrow ( E^2_{PK_{U_k}}( \mathbf{b}_k ),rk_{U^2_k \to U^1_k} )$,
        \item $E^1_{PK_J}( \mathbf{b}_k ) \leftarrow ( E^2_{PK_{U_k}}( \mathbf{b}_k ),rk_{U_k \to J} )$,
        \item Store $E^1_{PK_J} ( \mathbf{b}_k )$ into a set $\mathcal{F}$,
        \item $E^1_{PK_{U_k}}( \mathbf{D}_k(t) ) \leftarrow ( E^1_{PK_{U_k}}( \mathbf{E}(t)), E^1_{PK_{U_k}}( \mathbf{b}_k ),\mathbf{G})$,
        \item Extract the corresponding $\mathbf{c}$ from $\{ \mathbf{c}\}$,
        \item Extract the corresponding $\mathbf{B}^m$ from $\{ \mathbf{B}^m\}$,
        \item Send $E^1_{PK_{U_k}}( \mathbf{D}_k(t) )$, $\mathbf{c}$, and $\mathbf{B}^m$ to the $k$-th user.
      \end{itemize}
      
      \underline{$k$-th User:}
      \begin{itemize}
        \item $\mathbf{D}_k(t) \leftarrow ( E_{PK_{U_k}}^1( \mathbf{D}_k(t)),SK_{U_k})$,
        \item $\mathbf{m}^k \leftarrow ( \mathbf{c}, \mathbf{B}^m, \mathbf{D}_k )$.
      \end{itemize}} 
    \end{minipage}
  }
  \caption{{FairCMS-I: Media sharing (\textbf{Part 2}).}}
  \label{fig:3}
   \vspace{-2pt}
\end{figure}

\begin{figure}[!t]
  \centering
  
  \fbox{
    \begin{minipage} [!t]{0.46\textwidth}
      {\textbf{Input:} {\color{black}The suspicious media content ${\tilde{\mathbf{m}}}^k$, $\{ \bar{\mathbf{G}} \}$, the original media content $\mathbf{m}$,} and the set of fingerprints $\mathcal{F}$.}
      
      {\textbf{Output:} The judge finds the copyright violator.}
      
      {\textbf{Procedure:}}
      
      {\underline{Owner:} 
      \begin{itemize}
        \item Extract and download the corresponding $\bar{\mathbf{G}}$ from the cloud, 
        \item Download the set $\mathcal{F}$ from the cloud,
        \item {\color{black}Send $\mathbf{m}$, $\bar{\mathbf{G}}$, and $\mathcal{F}$ to the judge.}
      \end{itemize}
      
      \underline{Judge:}
      \begin{itemize}
        \item {\color{black}Decrypt the fingerprints in the set $\mathcal{F}$ with the private key $SK_J$ of the judge,}
        \item $\tilde{\mathbf{b}}_k \leftarrow ( \tilde{\mathbf m}^k,\bar{\mathbf{G}}, \mathbf{m})$,
        \item Compare $\tilde{\mathbf{b}}_k$ to the decrypted fingerprints in $\mathcal{F}$.
      \end{itemize}} 
    \end{minipage}
  }
  \caption{{\color{black}{FairCMS-I: Traitor tracing (\textbf{Part 3}).}}}
  \label{fig:4}
   \vspace{-2pt}
\end{figure}

{\color{black}We emphasize that \textbf{Part 1} is just an initialization step that only needs to be executed once before media sharing begins, while \textbf{Part 2} is executed upon requesting of media content copy from each authorized user. \textbf{Part 3} is only executed for each detected suspicious media content copy.} 
Therefore, in the case of a large number of users, the owner's overhead in \textbf{Part 2} should be the primary concern for a media sharing system. 
Fortunately, in our design, by delegating the operations securely to the cloud, {\color{black}now in \textbf{Part 2} the owner only needs to calculate and send a re-encryption key,} which only incurs negligible overhead. As a result, this scheme well solves the bottleneck caused by insufficient owner-side resources. Also, since \textbf{Part 2} is executed online for each user, this scheme clearly meets the scalability requirements. {The achievement of the three security goals of FairCMS-I will be discussed in detail in Section~\ref{Achieving}.}

In addition, it is noted that two re-encryption keys $rk_{U^2_k \to U^1_k}$ and $rk_{U_k \to J}$ are used in the scheme to protect the users' fingerprints. Among them, $rk_{U^2_k \to U^1_k}$ is used to change $E^2_{PK_{U_k}}(\mathbf{b}_k)$ into $E^1_{PK_{U_k}}(\mathbf{b}_k)$, and $rk_{U_k \to J}$ is used to change $E^2_{PK_{U_k}}(\mathbf{b}_k)$ into $E^1_{PK_J}(\mathbf{b}_k)$. 
In comparison, the method adopted in \cite{zhang2018you} is that the $k$-th user encrypts his/her fingerprint $\mathbf{b}_k$ with his/her public key $PK_{U_k}$ and judge's public key $PK_{J}$ respectively, and then sends the two results $E^1_{PK_{U_k}}(\mathbf{b}_k)$ and $E^1_{PK_J}(\mathbf{b}_k)$ to the cloud. 
However, in this case, an unfaithful user can evade traitor tracing by producing two different fingerprints, i.e., a $\mathbf{b}^{'}_k$ is used to produce $E^1_{PK_{U_k}}(\mathbf{b}^{'}_k)$ and another $\mathbf{b}^{''}_k$ is used to produce $E^1_{PK_J}(\mathbf{b}^{''}_k)$. 
Our proposed scheme avoids this vulnerability because the user only has freedom to generate one fingerprint $\mathbf{b}_k$, and $E^1_{PK_{U_k}}(\mathbf{b}_k)$ (for fingerprint embedding) and $E^1_{PK_J}(\mathbf{b}_k)$ (for traitor tracing) are re-encryption results of the same second-level ciphertext $E^2_{PK_{U_k}}(\mathbf{b}_k)$.

{\color{black}Moreover, FairCMS-I does not perform any processing on the encrypted media content stored in the cloud,} but only performs homomorphic operations and re-encryption operations on the encrypted LUT and fingerprint that are much smaller in size, {which results} in outstanding cloud-side efficiency. In contrast, {\color{black}the two schemes proposed by Zhang \textit{et al.} \cite{zhang2018you} both require homomorphic operations and re-encryption operations on the media content,} which obviously consumes a lot more cloud resources. This means that FairCMS-I can save users a large amount of cost of renting cloud resources, thus achieving one of the prominent advantages of our work.

Finally, we emphasize that the attack from the communication channel is not considered in the adopted threat model, so $\mathbf{B}^m$, $\mathbf{G}$, and $\bar{\mathbf{G}}$ are all transmitted in plaintext. In practical applications, if the communication channel does not meet this assumption, simple modifications can be made to the proposed scheme. On the one hand, we can use public key cryptography to encrypt the transmitted plaintext. For example, in \textbf{Part 1} the owner can encrypt $\mathbf{G}$ with the cloud's public key $PK_C$ before sending it to the cloud, who then decrypts $E_{PK_C}(\mathbf{G})$ with his/her own private key $SK_C$. {\color{black}On the other hand, the same $\mathbf{B}^m$ can be generated on different entities by sharing the session key $SK_m$, in the same way as in the AFP scheme. Meanwhile, bandwidth savings can also be obtained by transferring $SK_m$ instead of $\mathbf{B}^m$. Apart from the above, message authentication code and digital signature can also be added if integrity and source authentication are needed.}

\subsection{Scheme of FairCMS-II}
\label{sec:FairCMS}
{\color{black}The encryption and decryption on the media content $\mathbf{m}$ in FairCMS-I are essentially symmetric-key cipher although the E-LUT and D-LUT used are different,} so it is highly efficient since only the single-value additive operation is used in this cipher. 
{\color{black}However, the IND-CPA security of the encrypted media content collection $\{\mathbf{c}\}$ stored in the cloud cannot be achieved \cite{celik2008lookup}. In applications with high media privacy requirements,} a more secure cloud media sharing scheme is desired. {\color{black}In this concern, we propose the second cloud media sharing scheme FairCMS-II, in which all media content stored in the cloud are encrypted by the lifted-ElGamal based encryption scheme.} 

{\color{black}In FairCMS-II, the key, fingerprint, E-LUT, and encoding matrix are all generated in the same way as in FairCMS-I. In addition,} the calculation and management process for D-LUTs is also the same as in FairCMS-I, except that there is no need to distribute the resulting D-LUTs to users. That is, the cloud is not required to send the last calculated $E^1_{PK_{U_k}}(\mathbf{D}_k(t))$ (\textbf{Part 2} of FairCMS-I) to the $k$-th user, but instead storing it in a set $\mathcal{D}$. As will be shown in the following detailed discussion, this trick will enable the full usage of the computational power of the cloud and reduce bandwidth usage.

\noindent \textit{{\color{black}Cloud-Based Media Sharing:}}

{\color{black}The owner encrypts the media content collection $\{\mathbf{m}\}$ to be shared with his/her public key $PK_O=Z^{b_1}$, and then sends the resulting encrypted media content collection $\{E^2_{PK_O}(\mathbf{m})\}$ $=\{(g^r,Z^{\mathbf{m}}Z^{b_1r})\}$ to the cloud.} In addition, the owner use the session keys he/she selects to generate the binary matrix collection $\{\mathbf{B}^m\}$, which is then sent to the cloud.

The cloud calculates the secret matrix collection $\{ \bar{\mathbf{G}} \}$ with $\{\mathbf{B}^m\}$ and $\mathbf{G}$, and then stores $\{\bar{\mathbf{G}}\}$ and $\{\mathbf{B}^m\}$. 
{\color{black}In addition, the cloud calculates the encrypted ciphertext media content collection $\{E^2_{PK_O}(\mathbf{c})\}$ from $\{E^2_{PK_O}(\mathbf{m})\}$, $E^2_{PK_O}(\mathbf{E}(t))${,} and $\{\mathbf{B}^m\}$ by}
\begin{equation*}
E^2_{PK_O}(c_i)=E^2_{PK_O}(m_i)\cdot \prod\limits^{T - 1}_{t = 0} E^2_{PK_O}(\mathbf{E}(t))^{\mathbf{B}^m(i,t)},
\end{equation*}
where $c_i$ {(resp. $m_i$)} is the $i$-th element of $\mathbf{c}$ {(resp. $\mathbf{m}$)} and $i\in\{0,2,...,M-1\}$. {\color{black}We emphasize that this process can be accomplished successfully as the second-level ciphertext of the lifted-ElGamal based PRE scheme satisfies the property of additive homomorphism.} Afterwards, the cloud stores $\{E^2_{PK_O}(\mathbf{c})\}$.

Suppose the $k$-th user want to acquire one of the owner's media content $\mathbf{m}$. {\color{black}The cloud then extracts the corresponding encrypted ciphertext media content $E^2_{PK_O}(\mathbf{c})$ from the encrypted ciphertext media content collection.} Subsequently, the cloud converts $E^2_{PK_O}(\mathbf{c})$ into $E^1_{PK_{U_k}}(\mathbf{c})$ with the received re-encryption key $rk_{O \to {U_k}}$ using the PRE technique. Joint decryption and fingerprinting is then processed by the cloud in the ciphertext domain as follows:
\begin{equation*}
\begin{aligned}
E^1_{PK_{U_k}}({m_i^k} ) = E^1_{PK_{U_k}}({c_i}) 
\cdot \prod\limits_{t = 0}^{T - 1}E^1_{PK_{U_k}}(\mathbf{D}_k(t))^{\mathbf{B}^m(i,t)}.
\end{aligned}
\end{equation*}
{\color{black}The cloud then sends the resulting $E^1_{PK_{U_k}}(\mathbf{m}^k)$ to the $k$-th user, who decrypts $E^1_{PK_{U_k}}(\mathbf{m}^k)$ with his/her own private key $SK_{U_k}=a_2$ to obtain the media content $\mathbf{m}^k$ embedded with his/her personal fingerprint. We emphasize that the cloud cannot get the plaintext version $\mathbf{m}^k$ of the watermarked media content as the cloud does not have the private key $SK_{U_k}$ of the $k$-th user to decrypt the resulting $E^1_{PK_{U_k}}(\mathbf{m}^k)$.}

{\color{black}The whole FairCMS-II scheme is summarized as follows. First, suppose an owner rents the cloud's resources for media sharing, the owner and the cloud execute \textbf{Part 1} as shown in Fig. \ref{fig:5}. Then, suppose the $k$-th user makes a request indicating that he/she wants to access one of the owner’s media content $\mathbf{m}$, the entities execute \textbf{Part 2} after the $k$-th user is authorized by the owner as shown in Fig. \ref{fig:6}. Finally, the arbitration and traitor tracing process follows the same approach of FairCMS-I and is thus omitted here.}

\begin{figure}[!t]
  \centering
  
  \fbox{
    \begin{minipage} [!t]{0.468\textwidth}
      {\textbf{Input:} Security parameter $\lambda$.}
      
      {\textbf{Output:} The cloud stores $\{\bar{\mathbf{G}}\}$, $\mathbf{G}$, $E^2_{PK_O}(\mathbf{E}(t))$, $\{\mathbf{B}^m\}$, and $\{E^2_{PK_O}(\mathbf{c})\}$.}
      
      {\textbf{Procedure:}}
      
      {\underline{Owner:} 
      \begin{itemize}
        \item Generate $( \mathbf{E},\mathbf{G},PK_O,SK_O,\{SK_m\} )$,
        \item $E^2_{PK_O}( \mathbf{E}(t) ) \leftarrow ( \mathbf{E}(t),PK_O )$,
        \item $\{E^2_{PK_O}(\mathbf{m})\} \leftarrow ( \{\mathbf{m}\},PK_O)$ ,
        \item $\{t_{ih}\} \leftarrow \{SK_m\}$,
        \item $\{\mathbf{B}^m\} \leftarrow \{t_{ih}\}$,
        \item Send $\mathbf{G}$, $E^2_{PK_O}(\mathbf{E}(t))$, $\{\mathbf{B}^m\}$, and $\{E^2_{PK_O}(\mathbf{m})\}$ to the cloud.
      \end{itemize}
      
      \underline{Cloud:}
      \begin{itemize}
        \item $\{\bar{\mathbf{G}}\} \leftarrow ( \{\mathbf{B}^m\},\mathbf{G} )$,
        \item $\{E^2_{PK_O}(\mathbf{c})\} \leftarrow ( \{E^2_{PK_O}(\mathbf{m})\},E^2_{P{K_O}}(\mathbf{E}(t)),\{\mathbf{B}^m\} )$,
        \item Store $\{\bar{\mathbf{G}}\}$, $\mathbf{G}$, $E^2_{PK_O}(\mathbf{E}(t))$, $\{\mathbf{B}^m\}$, and $\{E^2_{PK_O}(\mathbf{c})\}$.
      \end{itemize}} 
    \end{minipage}
  }
  \caption{{FairCMS-II: Media storage (\textbf{Part 1}).}}
  \label{fig:5}
  \vspace{-2pt}
\end{figure}

\begin{figure}[!t]
  \centering
  
  \fbox{
    \begin{minipage} [!t]{0.468\textwidth}
      {\textbf{Input:} Security parameter $\lambda$, $\mathbf{G}$, $E^2_{PK_O}(\mathbf{E}(t))$, $\{\mathbf{B}^m\}$, and $\{E^2_{PK_O}(\mathbf{c})\}$.}
      
      {\textbf{Output:} The $k$-th user obtains the watermarked media content $\mathbf{m}^k$.}
      
      {\textbf{Procedure:}}
      
      {\underline{$k$-th User:} 
      \begin{itemize}
        \item Generate $( \mathbf{b}_k,PK_{U_k},SK_{U_k} )$,
        \item $E^2_{PK_{U_k}}( \mathbf{b}_k ) \leftarrow ( \mathbf{b}_k,PK_{U_k} )$,
        \item $rk_{{U_k^2} \to {U_k^1}} \leftarrow ( SK_{U_k},PK_{U_k} )$,
        \item $rk_{{U_k} \to J} \leftarrow ( SK_{U_k},PK_J )$,
        \item Send $E^2_{PK_{U_k}}( \mathbf{b}_k )$, $rk_{{U_k^2} \to {U_k^1}}$, and $rk_{{U_k} \to J}$ to the cloud.
      \end{itemize}
      
      \underline{Owner:}
      \begin{itemize}
        \item $rk_{O \to {U_k}} \leftarrow ( SK_{O},PK_{U_k} )$,
        \item Send $rk_{O \to {U_k}}$ to the cloud.
      \end{itemize} 
      
      \underline{Cloud:}
      \begin{itemize}
        \item $E^1_{PK_{U_k}}(\mathbf{E}(t)) \leftarrow ( E^2_{PK_O}(\mathbf{E}(t)),rk_{O \to {U_k}} )$,
        \item $E^1_{PK_{U_k}}( \mathbf{b}_k ) \leftarrow ( E^2_{PK_{U_k}}( \mathbf{b}_k ),rk_{{U_k^2} \to {U_k^1}} )$,
        \item $E^1_{PK_J}( \mathbf{b}_k ) \leftarrow ( E^2_{PK_{U_k}}( \mathbf{b}_k ),rk_{{U_k} \to J} )$,
        \item Store $E^1_{PK_J}( \mathbf{b}_k )$ into a set $\mathcal{F}$,
        \item $E^1_{PK_{U_k}}( \mathbf{D}_k(t) ) \leftarrow ( E^1_{PK_{U_k}}( \mathbf{E}(t) ), E^1_{PK_{U_k}}( \mathbf{b}_k ),\mathbf{G} )$,
        \item Extract the corresponding $E^2_{PK_O}(\mathbf{c})$ from $\{E^2_{PK_O}(\mathbf{c})\}$,
        \item Extract the corresponding $\mathbf{B}^m$ from $\{\mathbf{B}^m\}$,
        \item $E^1_{PK_{U_k}}(\mathbf{c}) \leftarrow ( E^2_{PK_O}(\mathbf{c}),rk_{O \to {U_k}} )$,
        \item $E^1_{PK_{U_k}}(\mathbf{m}^k) \leftarrow ( E^1_{PK_{U_k}}(\mathbf{c}),E^1_{PK_{U_k}}( \mathbf{D}_k(t) ), \mathbf{B}^m )$,
        \item Send $E^1_{PK_{U_k}}(\mathbf{m}^k)$ to the $k$-th user.
      \end{itemize}
      
      \underline{$k$-th User:}
      \begin{itemize}
        \item {$\mathbf{m}^k \leftarrow ( E^1_{PK_{U_k}}(\mathbf{m}^k), PK_{U_k} )$.}
      \end{itemize}} 
    \end{minipage}
  }
  \caption{{FairCMS-II: Media sharing (\textbf{Part 2}).}}
  \label{fig:6}
   \vspace{-2pt}
\end{figure}

{\color{black}The owner-side efficiency and scalability performance of FairCMS-II are directly inherited from FairCMS-I,} and the achievement of the three security goals of FairCMS-II is also shown in Section~\ref{Achieving}. Comparing to FairCMS-I, {\color{black}it is easy to see that in FairCMS-II the cloud's overhead is increased considerably due to the adoption of re-encryption operations and homomorphic operations on the ciphertext of the media content,} which means it will cost more for the owner on renting the cloud's resources. {\color{black}We regard this as the trade-off between security and cost. In actual use, the two proposed schemes can be selected according to different security requirements.} {The flexibility of choice in cloud-side efficiency also constitutes one of the prominent advantages of our work.}

We emphasize that the joint decryption and fingerprinting operation can also be transferred to users. {\color{black}That is, in \textbf{Part 2}, the cloud can send $\{E^1_{PK_{U_k}}(\mathbf{c})\}$ to the $k$-th user instead of calculating $E^1_{PK_{U_k}}(\mathbf{m}^k)$ in the ciphertext domain.} Subsequently, the $k$-th user needs to decrypt $\{E^1_{PK_{U_k}}(\mathbf{c})\}$ with his/her own private key, and then performs joint decryption and fingerprinting operation with his/her personal D-LUT, {\color{black}which can be obtained from the cloud as in FairCMS-I.} In this case, some of the overhead is transferred from the cloud to the user. In general, it is recommended that the cloud performs the joint decryption and fingerprinting operation in the ciphertext domain {so as to take full advantage of the power of the cloud.}

\section{Achieving the Security Goals}
\label{Achieving}
This section proves that both FairCMS-I and FairCMS-II can solve the three security/privacy problems with different privacy/efficiency trade-offs. 

\subsection{Problem 1: {Data Privacy Leakage and Access Control}}
\label{Achieving_Problem1}
{\color{black}As discussed in Section~\ref{sec:problemstatement}, the private data in this paper consists of the users' fingerprints, the LUTs, and the owner's media content.}

{\newtheorem{thm}{\bf Theorem}
\begin{thm}\label{thm1}
{\color{black}The users’ fingerprints and the LUTs in both schemes, as well as the owner's media content in FairCMS-II are privacy-assured against the chosen-plaintext attack by the cloud under the assumptions of extended Decisional Bilinear Diffie–Hellman (eDBDH) \cite{boneh2003identity} and discrete logarithm.}
\end{thm} 
}

{\begin{proof}
In FairCMS-I and FairCMS-II, the $k$-th user encrypts his/her fingerprint $\mathbf{b}_k$ into $E^2_{PK_{U_k}}(b_k)$ with his/her own public key $PK_{U_k}$, and then sends $E^2_{PK_{U_k}}(b_k)$ together with two re-encryption key $rk_{U_k^2 \to U_k^1}$ and $rk_{U_k \to J}$ to the cloud. As a result, different encrypted versions of the fingerprint accessed by the cloud are $E^2_{PK_{U_k}}(b_k)$, $E^1_{PK_{U_k}}(b_k)$, and $E^1_{PK_J}(b_k)$. 

We then check the privacy of LUTs, i.e., the E-LUT and D-LUTs. In the two proposed schemes, the owner shares the E-LUT $\mathbf{E}$ in the encrypted form  $E^2_{PK_O}( \mathbf{E}(t) )$ to the cloud who then utilizes the PRE technique to convert $E^2_{PK_O}( \mathbf{E}(t) )$ into $E^1_{PK_{U_k}}( \mathbf{E}(t) )$ with the re-encryption key  $rk_{O \to {U_k}}$. In addition, the cloud calculates $E^1_{PK_{U_k}}( \mathbf{D}_k )$ in the ciphertext domain with the homomorphic property, as shown in Eqs.~(\ref{Eq:enccalwkl}) and (\ref{Eq:enccaldkl}). To sum up, only $E^2_{PK_O}( \mathbf{E}(t) )$, $E^1_{PK_{U_k}}( \mathbf{E}(t) )$, and $E^1_{PK_{U_k}}( \mathbf{D}_k )$ are exposed to the cloud. 

{\color{black}By a similar analysis, the ciphertexts about media content available to the cloud are $E^2_{PK_O}(\mathbf{m})$, $E^2_{PK_O}(\mathbf{c})$, $E^1_{PK_{U_k}}(\mathbf{c})$, and $E^1_{PK_{U_k}}(\mathbf{m}^k)$. 
As a result, the problem of proving that the users' fingerprints, the LUTs, and the owner's media content are protected against the honest-but-curious cloud translates to the problem of proving that both first-level and second-level ciphertexts of the lifted-ElGamal based PRE scheme are all secure, which can be achieved by the reduction below.}

Let $\Pi$ denote the lifted-ElGamal based PRE scheme used in this paper and $\Pi'$ denote the native ElGamal based PRE scheme proposed in \cite{ateniese2006improved}, where the difference only lies in whether the message $m$ is lifted. Then we construct a probabilistic polynomial-time (PPT) adversary $\mathcal{A}'$ who attacks $\Pi'$ by using the PPT adversary $\mathcal{A}$ who attacks $\Pi$ as a subroutine. {\color{black}We further show if the subroutine $\mathcal{A}$ breaks the definition of indistinguishability of $\Pi$,} then the constructed adversary $\mathcal{A}'$ will break the indistinguishability of $\Pi'$, which is always false since $\Pi'$ is secure. 
The details of constructing $\mathcal{A}'$ is as follows.

\textbf{The experiment that $\mathcal{A}'$ attacks $\Pi'$:}
\begin{enumerate}[~~1.]
\item Run $\mathcal{A}(1^\lambda)$ to obtain a pair of messages $m_0$ and $m_1$.
\item Compute $m'_0=Z^{m_0}$ and $m'_1=Z^{m_1}$, then outputs a pair of messages $m'_0$ and $m'_1$.
\item Forward $m'_0$ and $m'_1$ to $\mathcal{A}'$'s encryption oracle $\mathcal{O'}$ and get $E^1_{PK}(m_b)$ and $E^2_{PK}(m_b)$, give them to $\mathcal{A}$.
\item Whenever $\mathcal{A}$ queries the encryption of other messages $m$, $\mathcal{A}'$ answers this query in the following way:
\begin{enumerate}[(a)]
\item Compute $m'=Z^{m}$ locally,
\item Query $\mathcal{O'}$ with $m'$ and obtain responses $E^1_{PK}(m)$ and $E^2_{PK}(m)$, then return them to $\mathcal{A}$.
\end{enumerate}
\item Continue answering encryption oracle queries of $\mathcal{A}$ as before until $\mathcal{A}$ gives its output. If $\mathcal{A}$ outputs 0, $\mathcal{A}'$ outputs $b'=0$;  otherwise, $\mathcal{A}'$ output $b'=1$.
\end{enumerate}

The output of the experiment is defined to be 1 if $b'=b$, and 0 otherwise. If ${\textnormal{PubK}}^{\textnormal{cpa}}_{\mathcal{A}',\Pi'}(\lambda)=1$, we say that $\mathcal{A}'$ succeeds. Note that the view of $\mathcal{A}$ when run as a subroutine of $\mathcal{A}'$ is distributed identically to the view of $\mathcal{A}$ when it interacts with $\Pi$ itself, and $\mathcal{A}'$ directly takes the final output of $\mathcal{A}$ as its final output $b'$. So
\begin{equation}
Pr[{\textnormal{PubK}}^{\textnormal{cpa}}_{\mathcal{A}',\Pi'}(\lambda)=1]=Pr[ \textnormal{PubK}^{\textnormal{cpa}}_{\mathcal{A},\Pi}(\lambda)=1].
\label{Eq:proofeq}
\end{equation}

{\color{black}According to the proof result by Ateniese \textit{et al.} \cite{ateniese2006improved}, if and only if the assumptions of eDBDH and discrete logarithm are established, the native ElGamal based PRE scheme meets the IND-CPA security, i.e.,}
\begin{equation*}
|\frac{1}{2}-Pr[{\textnormal{PubK}}^{\textnormal{cpa}}_{\mathcal{A}',\Pi'}(\lambda)=1]| \leq \textnormal{negl}(\lambda),
\end{equation*}
where $\textnormal{negl}(\lambda)$ denotes a negligible function parameterized by $\lambda$. Thus, based on Eq.~(\ref{Eq:proofeq}), we have
\begin{equation*}
|\frac{1}{2}-Pr[{\textnormal{PubK}}^{\textnormal{cpa}}_{\mathcal{A},\Pi}(\lambda)=1]| \leq \textnormal{negl}(\lambda),
\end{equation*}
which concludes the proof.
\end{proof}
}

{\color{black}We then turn to analyze the privacy protection performance of the owner's media content in FairCMS-I, where the owner encrypts the media content $\mathbf{m}$ based on ${\mathbf{B}^m}$ and $\mathbf{E}$ using the single-value alteration encryption, namely Eq.~(\ref{Eq:encryption}).} 
Although the IND-CPA security of $\mathbf{c}$ cannot be rigorously proved, this is not considered a major vulnerability for the following two reasons. For one thing, {\color{black}as long as the variance ${\sigma_E}^2$ of $\mathbf{E}$ is set large enough,} the statistical hiding can be achieved \cite{2011Efficient}, i.e., a PPT adversary cannot distinguish the statistical difference between two different encryption results. The security of the single-value alteration encryption can be further enhanced by computing $\mathbf{c}$ modulo an integer $R$ \cite{bianchi2014ttp}. In this way, as long as $R \ll \sigma_E$, regardless of the signal values in plaintext, the values of the encrypted signals are approximately uniformly distributed on $[0,R-1]$. In addition, the linear estimation attack on ciphertext $\mathbf{c}$ was also analyzed in \cite{celik2008lookup}, {\color{black}and the result shows that the attack quality is poor as long as the size $T$ of the LUTs (w.r.t. the size $M$ of the media content) is set large enough. For another, as also argued in \cite{celik2008lookup,xiao2023fingerchain,xiao2022preview}, it is not always necessary to perfectly guarantee the security of $\mathbf{c}$ in general scenarios (in other words, non-highly confidential scenarios), but rather, it would be sufficient to make cracking the encryption as difficult as removing the watermark from the watermarked media content. In FairCMS-I, the watermark is embedded into $\mathbf{m}$ based on Eq.~(\ref{Eq:decryption}), i.e., $\mathbf{m}_k=\mathbf{m}+\mathbf{B}_m\mathbf{W}_k$, which is identical in form with Eq.~(\ref{Eq:encryption}) used for media content encryption, i.e., $\mathbf{c} = \mathbf{m} + \mathbf{B}_m \mathbf{E}$. As a result, if an attacker succeeds in cracking $\mathbf{c}$ without knowing the E-LUT $\mathbf{E}$, then he/she can remove the embedded watermark from $\mathbf{m}_k$ in the same way without knowing the W-LUT $\mathbf{W}_k$ to obatin $\mathbf{m}$, and vice versa. Furthermore, it might be harder to crack $\mathbf{c}$ because the variance of $\mathbf{E}$ is much larger than that of $\mathbf{W}_k$. Therefore, we conclude that the protection of media content privacy in FairCMS-I is sufficient for general scenarios.}


{We finally analyze the effect of the two schemes on access control. For FairCMS-I, access control is achieved by having the owner distribute D-LUTs only to authorized users, which can be realized here because distributing D-LUTs requires the owner to give the re-encryption key $rk_{O \to {U_k}}$. Unauthorized users who do not own the D-LUT cannot access the media content because they cannot perform the media decryption calculation, {\color{black}namely Eq. (\ref{Eq:decryption}). Furthermore, users with D-LUTs are prevented from accessing the media content outside the authorized scope because different media contents correspond to different $\mathbf{B}^m$.} In fact, this access control mechanism is entirely inherited from the original AFP scheme \cite{bianchi2014ttp}. This mechanism is not resistant to collusion between users, in which case $\mathbf{B}^m$ can be obtained from other authorized users. FairCMS-II solves this flaw by having the cloud compute Eq. (\ref{Eq:decryption}) in the ciphertext domain and send only $E^1_{PK_{U_k}}(\mathbf{m}^k)$ to authorized users.}

\subsection{Problems 2 and 3: Owner’s Copyright and Users’ Rights}
{\color{black}In the absence of any collusion between the owner, users, and the cloud, the two proposed schemes naturally inherit the performance of the TTP-free AFP \cite{bianchi2014ttp} in protecting the media owner’s copyright and the users' rights.}

In FairCMS-I and FairCMS-II, on the one hand, although the user generates the fingerprint $\mathbf{b}_k$ on his/her own, the user is unable to know the sequence $\bar{\mathbf{G}}\mathbf{w}_k$ that embedded in the watermarked media content due to the secrecy of $\mathbf{G}$. In the case of non-collusion, the linear assessment attack is a commonly used attack strategy for users to remove their own watermarks from the watermarked content. In this regard, Celik \textit{et al.} \cite{celik2008lookup} proved through detailed theoretical analysis that this attack strategy hardly works as long as the power of the watermark (w.r.t. the signal power) and the size of the LUTs (w.r.t. the content size) are set large enough. A more difficult situation to prevent is \textbf{collusion} attack, where a coalition of dishonest users compare their respective D-LUTs with each other in the hope of obtaining an untraceable copy of the content. The anti-collusion performance of the AFP scheme used in this paper is similar to that of a generic spread-spectrum watermark \cite{celik2008lookup}, and at the same time, it is compatible with mainstream anti-collusion approaches. Following this idea, two anti-collusion solutions for the {LUT-based} AFP scheme were later proposed by Bianchi \textit{et al.} \cite{bianchi2015anticollusion}.

On the other hand, since the owner has no knowledge of the $k$-th user's fingerprint $\mathbf{b}_k$ (the user does not transmit any information to the owner other than to request authorization), the embedded $\bar{\mathbf{G}}\mathbf{w}_k$ is also unknown to the owner. In this regard, the owner alone has no ability to frame the $k$-th user. 

Afterwards, we consider the case where the owner \textbf{colludes} with the cloud, i.e., the owner aggregates the knowledge of the cloud to launch attacks. 

\begin{thm}\label{thm2}
{\color{black}Even if the owner colludes with the cloud, the knowledge he/she aggregates from the cloud still does not allow him to frame the users.} 
\end{thm} 

\begin{proof}
We start by analyzing the situation where the owner attempts to passively attack with information gathered from the cloud. There are two ways for the owner to frame the $k$-th user, one is to embed the $k$-th user's fingerprint into any media content, which requires the knowledge of $\mathbf{b}_k$ or $\mathbf{D}_k$, and the other is to directly obtain $\mathbf{m}_k$. However, as \textbf{Theorem~\ref{thm1}} shows, the cloud cannot obtain any plaintext information about the $k$-th user's personal fingerprint and D-LUT in both schemes. In fact, the cloud also cannot obtain any plaintext information about the watermarked content $\mathbf{m}^k$ of the $k$-th user, because in FairCMS-I, the joint decryption and fingerprinting process occurs locally to the user, and in FairCMS-II, the security of $E^1_{PK_{U_k}}(\mathbf{m}^k)$ is guaranteed by \textbf{Theorem~\ref{thm1}}. {\color{black}Therefore, passive attacks cannot work even with the aggregated knowledge of the owner and the cloud.} 

We then proceed to analyze the situation of active attack. In fact, the owner cannot launch any effective active attacks (the cloud does not cooperate due to the assumption honest-but-curious), because the direct interaction between the owner and the user in both schemes is strictly limited (there is only one round in the authorization phase), and any other interaction request trying to obtain additional data will be deemed illegal by the user and abort the sharing process.
\end{proof}

\section{Efficiency Analysis}
\label{Sec:EfficiencyAnalysis}

In this section, we evaluate the efficiency of FairCMS-I and FairCMS-II and compare the results with the efficiency of user-side embedding AFP to demonstrate that these two schemes can be regarded as a secure outsourcing of AFP.

{\color{black}To make a fair comparison, we assume that the computational complexity of modular exponentiation, bilinear pairing and discrete logarithm are $O(\alpha)$, $O(\beta)$ and $O(\gamma)$, respectively. Moreover,} we assume that the computational overhead caused by plain-field multiplication and addition operations is negligible compared to that of the encrypted domain. {\color{black}In addition, the data expansion rate resulting from PRE encryption is $\delta$, i.e., for the media content $\mathbf{m}$ of size $M$,} its ciphertext is of size $O(\delta M)$. We further assume that the number of the users is $K$. For ease of presentation, {\color{black}the additive homormorphic encryption employed in the user-side embedding AFP \cite{bianchi2014ttp} is assumed to be the lifted-ElGamal encryption.}
The analyzed results are summarized in Table~\ref{tab:efficiency}, in which three cost cases including the computation cost, the communication cost, and the storage cost are considered. 

\begin{table*}[ht]
\centering
\setlength{\tabcolsep}{3.5mm}
\renewcommand\arraystretch{1.2}
\caption{Efficiency Analysis Results}
\label{my-Table1}
\begin{threeparttable}
\begin{tabular}{c|c|c|c|c}
\hline
\multicolumn{2}{c|}{} & FairCMS-I & FairCMS-II & AFP    
\\ \hline
\multirow{3}{*}{Computional Cost} & Owner & $O(K\alpha)$ & $O(K\alpha)$ & $O(K(L+1)T\alpha)$ 
\\ \cline{2-5} 
& Cloud & $O(KTL\alpha+K(T+2L)\beta)$ & $O(K(L+M)T\alpha+K(T+2L+M)\beta)$ & {\color{black}$\circ$\tnote{*}}  
\\ \cline{2-5} 
& User  & $O((T+L+2)\alpha+T\gamma)$ & $O((M+L+2)\alpha+M\gamma)$ & $O((T+L)\alpha+T\gamma)$                          
\\ \hline
\multirow{3}{*}{Communication Cost} & Owner & $O(M+T(L+\delta)+K)$ & $O(\delta M+T(L+\delta)+K)$ & $O(K(\delta L+\delta T+M))$ 
\\ \cline{2-5} 
& Cloud & {\color{black}$+$\tnote{*}} & {\color{black}$+$\tnote{*}} & {\color{black}$\circ$\tnote{*}}     
\\ \cline{2-5} 
& User  & $O(\delta(T+L)+M+2)$ & $O(\delta(M+L)+2)$ & $O(\delta(L+T)+M)$
\\ \hline
\multirow{3}{*}{Storage Cost} & Owner & $O(M)$ & $O(M)$ & $O(M(L+2)+K\delta L)$     
\\ \cline{2-5} 
& Cloud & $O(M(L+1)+K\delta L)$  & $O(M(L+\delta)+K\delta L)$  & {\color{black}$\circ$\tnote{*}}
\\ \cline{2-5} 
& User  & {\color{black}$\circ$\tnote{*}}  & {\color{black}$\circ$\tnote{*}}  & {\color{black}$\circ$\tnote{*}}              
\\ \hline
\end{tabular}
\begin{tablenotes}
\footnotesize
{\color{black}\item[*] `$+$' means that the communication cost of the cloud is the sum of the communication costs of the owner and all $K$ users. `$\circ$' means that the storage cost is zero.}
\end{tablenotes}
\end{threeparttable}
\label{tab:efficiency}
\vspace{-2pt}
\end{table*}

The computation cost is analyzed first. 
For both FairCMS-I and FairCMS-II, the cost of \textbf{Parts 1} and \textbf{3} is unconcerned, as they are performed offline only before the media sharing trail begins and when a copyright dispute occurs, respectively. 
Conversely, \textbf{Part 2}, which requires online interaction among the owner, users and the cloud, is the main burden of the operational efficiency of media sharing protocol.
In \textbf{Part 2} of FairCMS-I, the owner only costs $O(K\alpha)$ to produce re-encryption keys for all $K$ users. 
The cloud needs to spend $O(K(T+2L)\beta)$ to perform the re-encryption operations and spend $O(KTL\alpha)$ to calculate the D-LUTs in the ciphertext domain, where $T$ is the length of the LUTs and $L$ is the length of the users' fingerprints. 
The user needs to spend $O((L+2)\alpha)$ to encrypt his/her fingerprint and calculate the re-encryption keys, and spend $O(T\alpha+T\gamma)$ to decrypt his/her personalized D-LUT.
Similarly, in \textbf{Part 2} of FairCMS-II, the case of the owner’s side is the same as that of FairCMS-I, so the computation cost is still $O(K\alpha)$. 
{\color{black}In FairCMS-II, the cloud’s cost increases $O(KMT\alpha+KM\beta)$, and the cost to the user changes to $O((M+L+2)\alpha+M\gamma)$, where $M$ is the size of the required media content.} This is because in FairCMS-II, the joint decryption and fingerprinting operation is outsourced to the cloud who carries out all computations in the ciphertext domain. In the client-side embedding AFP scheme, the owner is responsible for encrypting the E-LUT and computing the D-LUTs online for each user in the encryption domain, which costs $O(K(L+1)T\alpha)$. The user needs to encrypt his/her own fingerprint and decrypt the encrypted D-LUTs, which costs $O((T+L)\alpha+T\gamma)$. 

{\color{black}The communication cost is then analyzed. In FairCMS-I, the owner sends the LUT encrypted media content $\mathbf{c}$ to the cloud with a communication cost of $O(M)$.} The user receives the public-key encrypted D-LUT and the LUT encrypted media content sent from the cloud, which costs $O( \delta T+M)$. We ignore the overhead of transmitting {$\mathbf{B}^{m}$} because it can be generated locally on all parties by sharing the same session key $SK_m$. {\color{black}In FairCMS-II, the media content that the owner deliver is encrypted with the public-key encryption, hence the communication cost is $O( \delta M)$.} Similarly, the user receives the fingerprint embedded media content encrypted with public-key cryptography at a cost of $O(\delta M)$. In addition, the owner needs to transmit $\mathbf{G}$, $E_{P{K_O}}^{2}({\mathbf{E}}(t))$ and the re-encryption keys in both schemes, which costs $O(T(L+\delta)+K)$. In the proposed two schemes, the user also needs to spend $O(\delta L+2)$ to send the encrypted fingerprint and the re-encryption keys. {\color{black}Note that the communication overhead of the cloud is the sum of the communication overheads of the owner and all $K$ users, which we mark as ``$+$" in Table~\ref{tab:efficiency}. Accordingly, in the AFP scheme, the communication costs of the owner and the user are $O(K(\delta L+\delta T +M))$ and $O( \delta (L+T) + M)$, respectively.}

At last, the storage cost is analyzed. {\color{black}We consider only the storage cost of media  contents and proof materials, including the original media content, the encrypted media content, the secret matrices, as well as the encrypted users' fingerprints. In the user-side embedding AFP scheme, these materials and contents are all stored by the owner, which costs $O(M(L$ $+2)+K\delta L)$. In the proposed schemes, the owner only costs $O(M)$ to store the original media content,} by transferring the remaining storage requirements to the cloud. Correspondingly, in FairCMS-I and FairCMS-II, the cloud requires the storage overhead of $O(M(L+1)+K\delta L)$ and $O(M(L+\delta)+K\delta L)$ respectively. {\color{black}The reason for the difference in storage overhead lies in the different encryption forms of the stored encryption media content. The user does not participate in the management of the proof materials and the encrypted media content,} thereby the storage overhead is zero, which is marked as ``$\circ$" in Table~\ref{tab:efficiency}.

The results in Table~\ref{tab:efficiency} confirms the two assertions we made earlier. For one thing, compared with the AFP scheme, it is easy to see from Table~\ref{tab:efficiency} that the owner-side computing, communication, and storage costs are reduced in both proposed schemes. 
In the AFP scheme, the cost of the owner increases linearly as the number of users $K$ increases. 
However, in FairCMS-I and FairCMS-II, this linear-increase relationship is eliminated, i.e., the owner merely prepares an encryption key for the newly joined user.
{This is particularly important, as we mentioned in Section~\ref{subsec:afp}, because only in this way, the increase in the number of users will not overload the cost of the owner, which is a major problem faced in the AFP scheme.} In other words, the owner can achieve significant cost savings by renting the cloud's resources, thereby enabling large-scale media sharing. 
Moreover, the larger the scale of media sharing, the more computing savings the owner can get, which fits exactly with the idea of computing outsourcing. Based on the above reasons, we conclude that both FairCMS-I and FairCMS-II can be regarded as the privacy-preserving outsourcing of user-side embedding AFP. 

For another, since $M>T>L$ and ${\delta}>1$, it is {intuitive} from Table~\ref{tab:efficiency} that in FairCMS-II, the cloud has increased in computing and storage costs compared to that of FairCMS-I. {\color{black}In fact, the cloud-side communication cost of also increases in FairCMS-II as the uniqueness of each media content copy used for distribution hinders the use of network-level bandwidth saving mechanisms such as multicasting and caching. This means that FairCMS-I is more efficient at the cost of lacking IND-CPA security.} Therefore, the trade-off between security and efficiency is evident.

Finally, we analyze the necessary efforts required from the user to obtain the protected content. As can be seen from Table~\ref{tab:efficiency}, {\color{black}compared with the AFP scheme, the computational cost and communication cost of the user in FairCMS-I only increase constant terms $O($2$\alpha)$ and $O($2$)$ respectively, which is negligible.} FairCMS-II has an increase in user-side overhead compared to FairCMS-I, because in FairCMS-II, the user transmits and decrypts the media content while in FairCMS-I, the user transmits and decrypts the D-LUTs. Nevertheless, this 
user-side efficiency level is substantially the same as the current mainstream media sharing schemes \cite{zhang2018you,frattolillo2019multiparty,xia2016privacy,frattolillo2016buyer}, {\color{black}which all require users to decrypt the received encrypted media content with their own private key.}

\section{Experiment}
\label{sec:simulation}
{\color{black}In this section, we conduct simulation experiments to evaluate the performance of the two proposed schemes in terms of visual effect, perception quality, success tracing rate, and efficiency.}

\subsection{Parameter Setting}
\label{subsec:par}
{\color{black}We use four gray scale images as the shared media content (excluding the experiments in Figs. \ref{fig13}(b), \ref{fig14}, and \ref{fig15}(b), as well as Table \ref{my-Table5}), each with a resolution of $512 \times 512$ pixels.} As shown in Fig.~\ref{fig:aspecost1}, the contents of the four images are Baboon, Pirate, Lena, and Car, respectively. For each image, the vector ${\mathbf{m}}$ is obtained by performing $8 \times 8$ discrete cosine transform (DCT) on the image and arranging the DCT coefficients in zig-zag order. We set the variance ${\sigma_E}^2$ of a E-LUT as $10^6$, the length $T$ of LUTs as $1000$, the length $L$ of fingerprints as $50$, the parameter $S$ as $4$, and the number $K$ of users as $200$. {\color{black}Moreover, to support encryption algorithm and homomorphic operations, E-LUT, $\sigma_W$, and ${\mathbf{m}}$ is quantized via the fixed point representation manner \cite{bianchi2014ttp}.} Specifically, only the E-LUT and $\sigma_W$ are quantized in the AFP scheme and FairCMS-I, and ${\mathbf{m}}$ is additionally quantized in FairCMS-II. In the experiments, the quantization is done with $15$ bits for the magnitude part and $5$ bits for the fractional part. {The performance test of watermarking is implemented by Matlab, and the efficiency test of cryptography is implemented by JavaScript and Python.} The experimental equipment is a desktop computer with an Intel Core i$7-8700$ processor and $16$ GB of RAM.

\begin{figure}[ht]
\centering
\begin{minipage}{0.24\linewidth}
\vspace{3pt}
\centerline{\includegraphics[width=\textwidth]{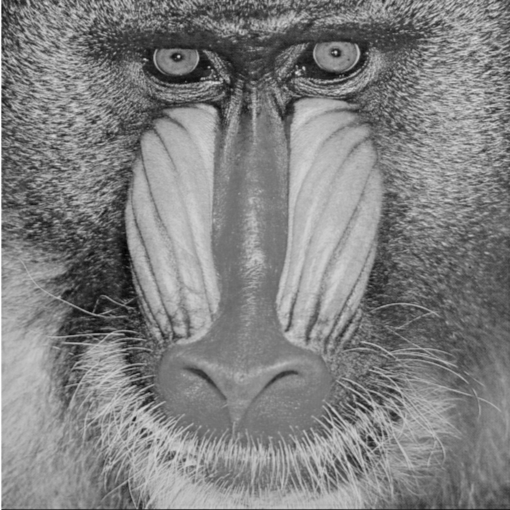}}
\centerline{(a)}
\end{minipage}
\begin{minipage}{0.24\linewidth}
\vspace{3pt}
\centerline{\includegraphics[width=\textwidth]{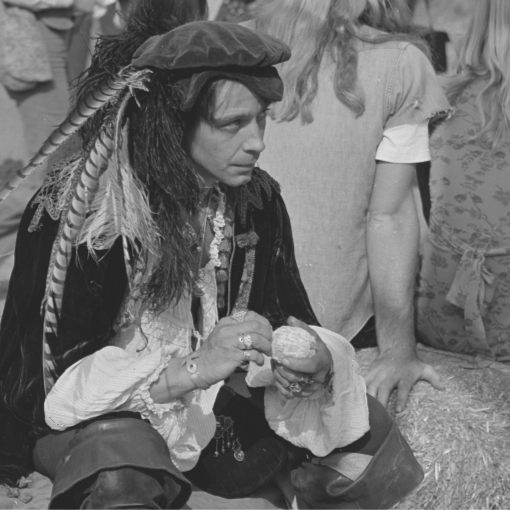}}
\centerline{(b)}
\end{minipage}
\begin{minipage}{0.24\linewidth}
\vspace{3pt}
\centerline{\includegraphics[width=\textwidth]{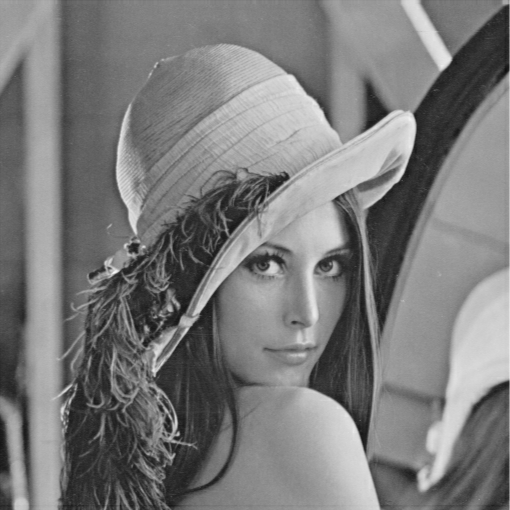}}
\centerline{(c)}
\end{minipage}
\begin{minipage}{0.24\linewidth}
\vspace{3pt}
\centerline{\includegraphics[width=\textwidth]{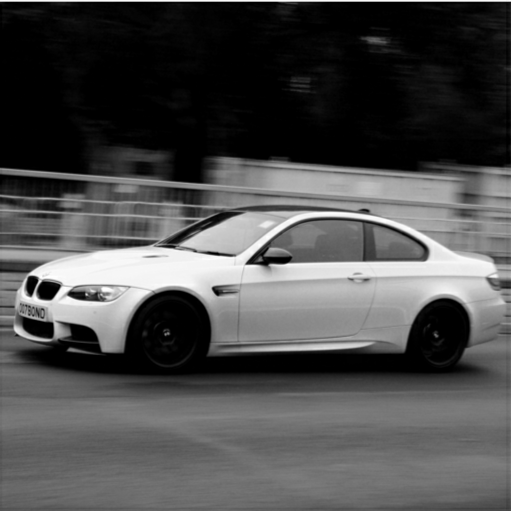}}
\centerline{(d)}
\end{minipage}
\caption{Four test images with a resolution of $512 \times 512$ pixels. The contents of (a)-(d) are Baboon, Pirate, Lena, and Car, respectively.}
\label{fig:aspecost1}
\vspace{-2pt}
\end{figure}

\subsection{{\color{black}Visual Effect}}
{\color{black}We first observe the effect of single-value alteration encryption and fingerprint embedding from a visual point of view.} We set the standard deviation ${\sigma_W}$ of W-LUT as $0.3$ and use  Fig.~\ref{fig:aspecost1}(a) to simulate the AFP scheme. From this, we obtain the corresponding E-LUT-encrypted image and the fingerprint embedded image, as shown in Fig.~\ref{fig:aspecost2}. The peak signal-to-noise ratio (PSNR) of the encrypted picture is $5.751$dB, and it can be seen from Fig.~\ref{fig:aspecost2}(a) that the single-value alteration encryption does not reveal any visible information of the original image. The PSNR of the decrypted and fingerprint embedded image is $50.952$dB, and its difference to the original image is not visible for naked eyes. Therefore, with appropriate selection of system parameters, AFP does not cause perception quality concerns.

\begin{figure}[ht]
\centering
\begin{minipage}{0.3\linewidth}
\vspace{3pt}
\centerline{\includegraphics[width=\textwidth]{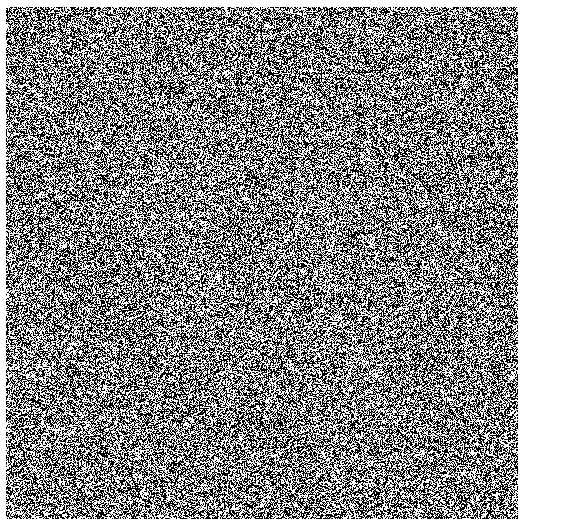}}
\centerline{(a)}
\end{minipage}
\begin{minipage}{0.3\linewidth}
\vspace{3pt}
\centerline{\includegraphics[width=\textwidth]{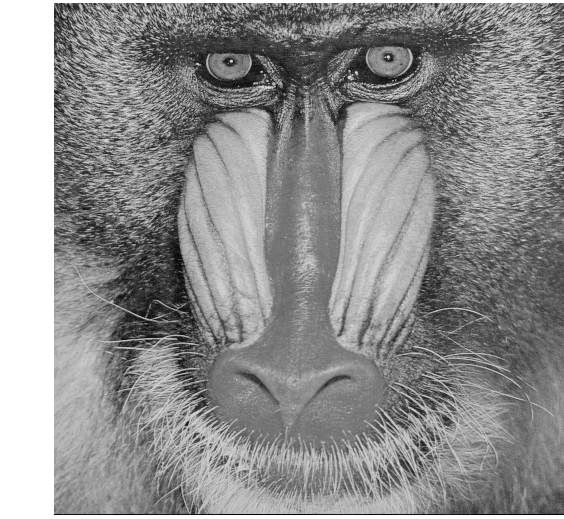}}
\centerline{(b)}
\end{minipage}
\caption{The encrypted image (PSNR $= 5.751$dB) and fingerprint embedded image (PSNR $= 50.952$dB) corresponding to Fig. \ref{fig:aspecost1}(a).}
\label{fig:aspecost2}
\vspace{-2pt}
\end{figure}

\begin{table}[htbp]
\centering
\setlength{\tabcolsep}{10mm}
\renewcommand\arraystretch{1.6}
 \scriptsize
\caption{Comparison of the PSNR (in dB) against different ${\sigma_W}$}
\label{my-Table2}
\setlength{\tabcolsep}{0.78mm}
\begin{tabular}{c|c|c|c|c|c|c}
\hline
  Method & ${\sigma_W}$=0.1 & ${\sigma_W}$=0.12 & ${\sigma_W}$=0.15 & ${\sigma_W}$=0.2 & ${\sigma_W}$=0.25 & ${\sigma_W}$=0.3 \\
 \hline
AFP/FairCMS-I & 67.895 & 61.150 & 57.512 & 55.385 & 52.843 & 51.244 \\ 
   \hline
 FairCMS-II & 68.009 & 61.183 & 57.538 & 55.401 & 52.866 & 51.241 \\
  \hline
\end{tabular}
\vspace{-2pt}
\end{table}

\begin{figure}[ht]
\centering
\includegraphics[scale=0.5]{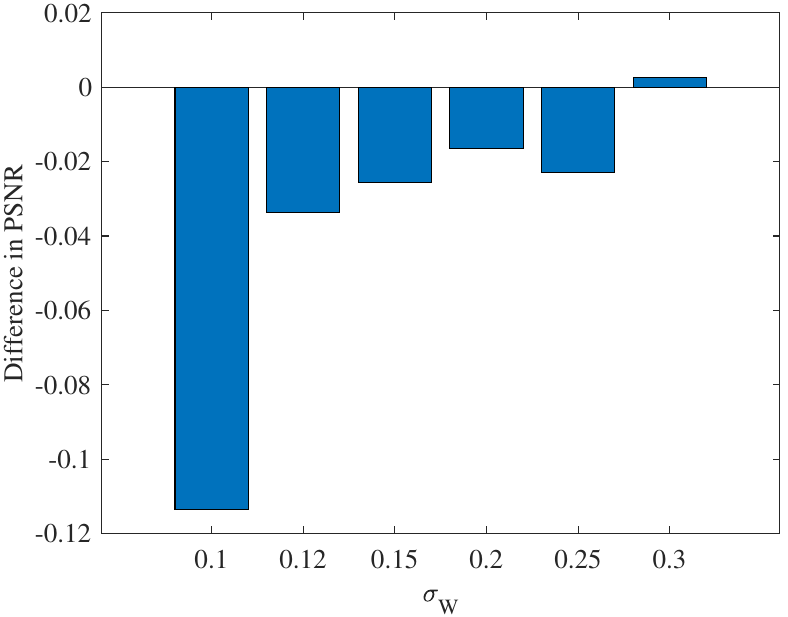}
\caption{The difference between AFP/FairCMS-I and FairCMS-II in PSNR.}
\label{fig9}
\vspace{-2pt}
\end{figure}

\begin{figure}[ht]
\centering
\includegraphics[scale=0.5]{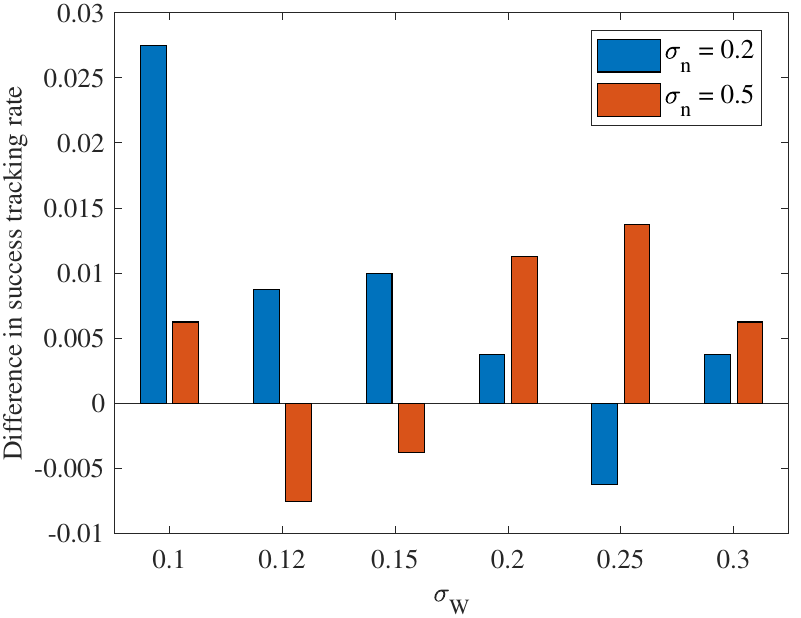}
\caption{The difference between AFP/FairCMS-I and FairCMS-II in success tracing rate.}
\label{fig10}
\vspace{-2pt}
\end{figure}

\subsection{Perception Quality}
\label{subsec:pquality}
Because the standard deviation ${\sigma_W}$ of W-LUT determines the intensity of fingerprint embedding, it directly affects the perception quality of the embedded image. Here we fix other parameters, such as the quantization bits, and investigate the effect of different values of ${\sigma_W}$ on the perception quality of the embedded image, which is again measured by the value of PSNR. 
{\color{black}In FairCMS-I, D-LUTs are computed in the ciphertext domain homomorphically, while the media encryption operation as well as the joint decryption and fingerprinting opera-tion are completed in the plaintext domain, which is exactly the same as in the AFP scheme.} As a result, the performance of FairCMS-I and the AFP scheme is theoretically consistent in terms of perception quality and traitor tracing. 
Table \ref{my-Table2} shows the comparison results between AFP/FairCMS-I and FairCMS-II in terms of perception quality. The experimental results are averaged over four tests, each using the image in Fig.~\ref{fig:aspecost1} in turn. As can be seen from Table \ref{my-Table2}, the smaller ${\sigma_W}$ is, the better the image perception quality. Fig. \ref{fig9} visualizes the difference between AFP/FairCMS-I and FairCMS-II in PSNR. As shown in Fig. \ref{fig9}, the PSNR values of FairCMS-II are slightly better than that of AFP/FairCMS-I under most settings. Compared with Fig. \ref{fig:aspecost2}, it can be concluded that all the listed PSNR values are high enough for usage in practical applications. Note that the watermark embedding here is not limited to working in the DCT domain. A suitable transform domain can be chosen according to the media type, such as the dual-tree complex wavelet transform domain for video \cite{huan2021exploring} and the Fourier transform domain for audio \cite{salah2021fourier}, {\color{black}to maximize the perception quality of the media content.}

\subsection{Success Tracing Rate}
\label{subsec:suc}
With other parameters such as the number of quantization bits being fixed, the success rate of traitor tracing is dominated by fingerprint embedding strength and noise power. We assume that the additive noise ${\mathbf{n}}$ obeys a Gaussian distribution with variance ${\sigma_n}^2$ and mean $0$. The matched filter decoder is selected for fingerprint detection since it has a fairly obvious suppression effect on Gaussian noise. Here we investigate the success rate of traitor tracing of AFP/FairCMS-I and FairCMS-II under different values of ${\sigma_W}$ and ${\sigma_n}$. Table \ref{my-Table3} presents the average results of four tests as in the previous subsection. In each test, we consider copyright infringement of each user and track that user. If a user's fingerprint is extracted without any errors, it is counted as a successful trace, otherwise it is counted as a failed trace. The success tracing rate is recorded as the number of successful traces divided by the total number of traces. {\color{black}As presented in Table \ref{my-Table3}, the increase of $\sigma_W$ and $\sigma_n$ has a positive and negative effect on the successful tracing rate as a whole, respectively.} Fig. \ref{fig10} visualizes the difference between AFP/FairCMS-I and FairCMS-II in success tracing rate. {\color{black}As shown in Fig. \ref{fig10}, compared with AFP/FairCMS-I, the success tracing rate of FairCMS-II decreases slightly in most cases. Nevertheless, this does not bother as the decrease is basically less than $0.015$.} 

\begin{table*}[htbp]
\centering
\caption{Comparison of the success tracing rate against different ${\sigma_n}$ and ${\sigma_W}$}
\label{my-Table3}
\setlength{\tabcolsep}{1.3mm}
\renewcommand\arraystretch{1.4}
\begin{tabular}{c|c|c|c|c|c|c|c|c|c|c|c|c}
\hline
\multirow{2}{*}{Method} & \multicolumn{6}{c|}{${\sigma_n}$=0.2}&\multicolumn{6}{c}{${\sigma_n}$=0.5}\\
\cline{2-13}
& ${\sigma_W}$=0.1 & ${\sigma_W}$=0.12 & ${\sigma_W}$=0.15 & ${\sigma_W}$=0.2 & ${\sigma_W}$=0.25 & ${\sigma_W}$=0.3 & ${\sigma_W}$=0.1 & ${\sigma_W}$=0.12 & ${\sigma_W}$=0.15 & ${\sigma_W}$=0.2 & ${\sigma_W}$=0.25 & ${\sigma_W}$=0.3 \\
\hline
AFP/FairCMS-I & 0.4888 & 0.6463 & 0.8775 & 0.9525 & 0.9775 & 0.9788 & 0.3275 & 0.6400 & 0.8488 & 0.9475 & 0.9775 & 0.9800\\ 
\hline
FairCMS-II & 0.4613 & 0.6375 & 0.8675 & 0.9488 & 0.9838 & 0.9750 & 0.3213 & 0.6475 & 0.8525 & 0.9363 & 0.9638 & 0.9738\\
\hline
\end{tabular}
\vspace{-2pt}
\end{table*}

\begin{figure}[ht]
\centering
\includegraphics[scale=0.5]{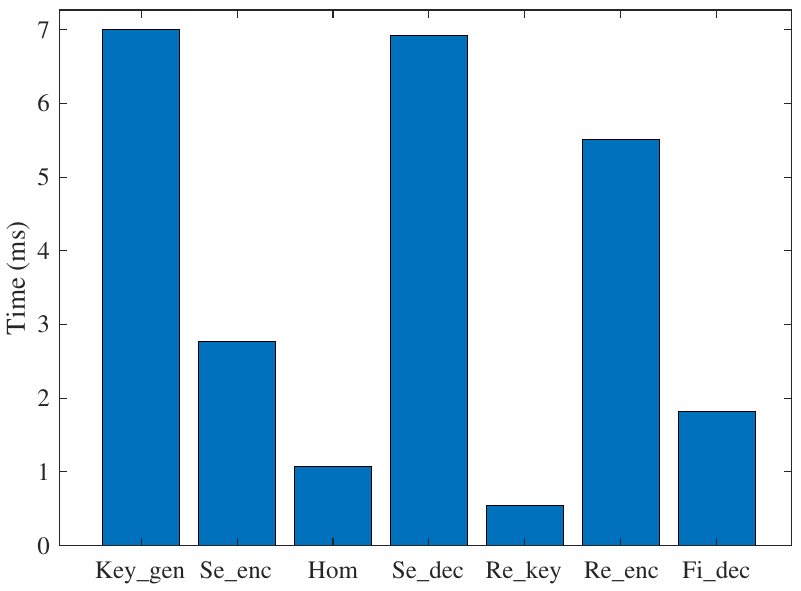}
\caption{{{\color{black}The computation time taken to perform each algorithm.} Key{\_}gen, Se{\_}enc, Hom, Se{\_}dec, Re{\_}key, Re{\_}enc, and Fi{\_}dec refer to the time consumed by key generation, second-level encryption, homomorphic operation, second-level decryption, generation of re-encryption key, re-encryption, and first-level decryption, respectively.}}
\label{fig11}
\vspace{-2pt}
\end{figure}

\begin{figure}[ht]
\centering
\includegraphics[scale=0.4]{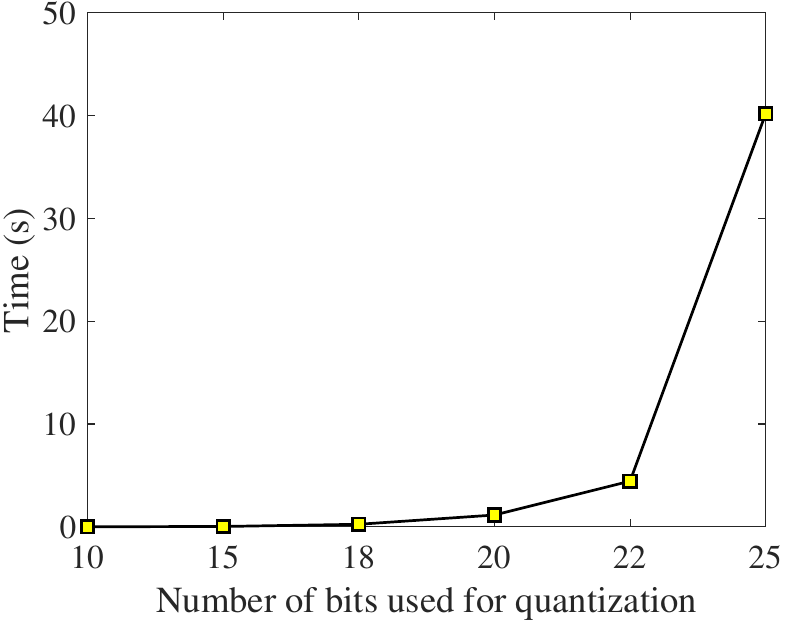}
\caption{{The time taken to solve discrete logarithm using the lookup table method \cite{CHEN2022} when different bits are used to quantize the plaintext.}}
\label{fig12}
\vspace{-2pt}
\end{figure}

\subsection{{Efficiency}}
\label{subsec:eff}
{\color{black}First, to evaluate the efficiency of implementing privacy-protected access control,} we test the time cost of each step in the lifted-ElGamal based PRE, and the results are presented in Figs. \ref{fig11} and \ref{fig12}. In the experiment, we adopt the lookup table method \cite{CHEN2022} to solve the discrete logarithm problem. {\color{black}In this method, a hash table that records all $(Z^m,m)$ pairs needs to be created first.} In this way, $m$ can be recovered from $Z^m$ only by querying the hash table during decryption, and the time consumed by the querying operation is negligible. The main time overhead lies in the establishment of the hash table. Fig. \ref{fig12} shows the time consumed in the establishment of the hash table when the plaintext is quantized by different bits. As can be seen from Fig. \ref{fig12}, when the plaintext is quantized to $20$ bits, the time cost is less than $2$ seconds. {\color{black}Note that the hash table only needs to be created once no matter how much data needs to be encrypted and decrypted.} As a result, solving the discrete logarithm problem has little effect on the efficiency of FairCMS-I and FairCMS-II. Fig. \ref{fig11} shows the efficiency performance of other steps except solving discrete logarithm. {\color{black}As can be seen from Fig. \ref{fig11}, all the steps can be completed in milliseconds and the relatively more time-consuming steps are those that require calculating bilinear pairings, including key generation, second-level decryption, and re-encryption.}

\begin{figure}[ht]
\centering
\begin{minipage}{0.7\linewidth}
\vspace{3pt}
\centerline{\includegraphics[width=\textwidth]{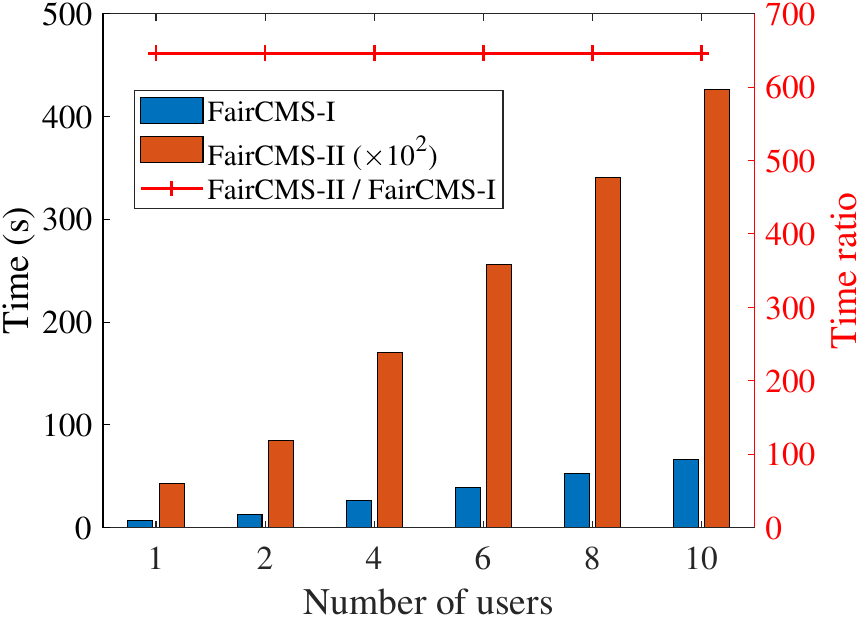}}
\centerline{(a)}
\end{minipage}
\begin{minipage}{0.7\linewidth}
\vspace{3pt}
\centerline{\includegraphics[width=\textwidth]{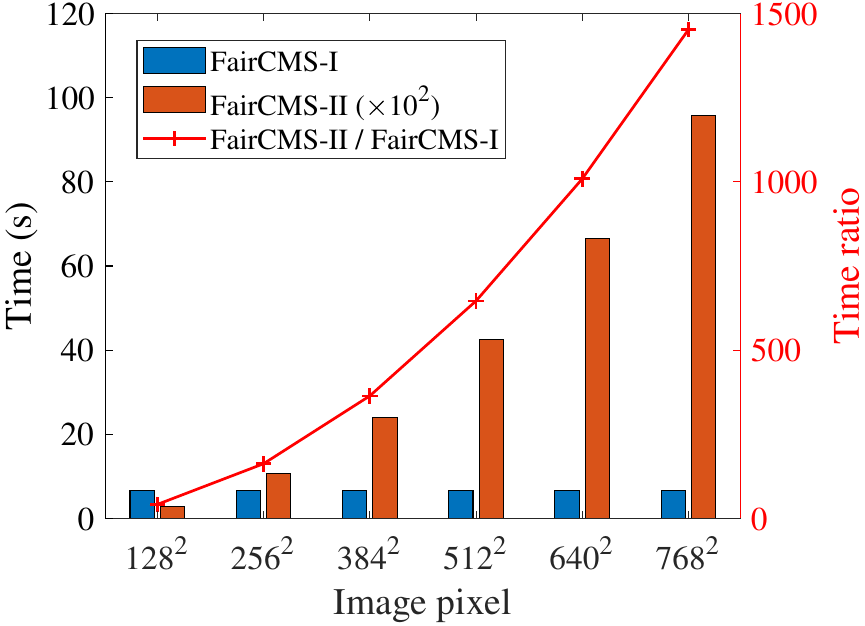}}
\centerline{(b)}
\end{minipage}
\caption{{\color{black}The comparison of cloud-side computational efficiency between FairCMS-I and FairCMS-II. The bars and polyline correspond to the left and right Y-axes, respectively. The time consumed by FairCMS-II is 100 times the reading on the Y-axis. (a) Efficiency comparison under different number of users. (b) Efficiency comparison under different image pixels.}}
\label{fig13}
\vspace{-2pt}
\end{figure}

{\color{black}Second, we compare the cloud-side efficiency of FairCMS-I and FairCMS-II, and the results are presented in Fig. \ref{fig13}. As shown therein, the cloud-side efficiency of FairCMS-I is significantly higher than that of FairCMS-II, thus validating our analysis in Section \ref{Sec:EfficiencyAnalysis}. The main reason for the cloud-side efficiency gain of FairCMS-I lies in the use of lightweight single-value alteration method to encrypt the media content, as shown in Fig. \ref{fig14}. This is the key to ensuring that the system is efficient when the size of the media content being shared (e.g., vedio) is large. The time cost of encrypting a video using the single-value alteration method is depicted in Table \ref{my-Table5}, and it is shown to be acceptably low. Therefore, we suggest that owners select FairCMS-I when the media content size is large and the security requirements is not excessively rigorous. In spite of this, there are no fixed thresholds for media content size and security requirements that can be used as a basis for recommendations regarding scheme selection.}

{\color{black}Finally, we conduct a comparative experiment to evaluate the proposed schemes against their relevant existing counterparts, and the results are displayed in Fig. \ref{fig15}. For FairCMS-I and FairCMS-II, we measure the time overhead of Part 2 as it is executed once for each user. For the other schemes, we evaluate their primary sources of overhead, namely encryption, homomorphic operations, and re-encryption operations, which also require execution once for each user. The lifted-ElGamal scheme is utilized in the evaluation for uniform comparison. As shown in Fig. \ref{fig15}, for one thing, FairCMS-I and FairCMS-II allow the owner to achieve significant computational overhead savings relative to the original AFP scheme \cite{bianchi2014ttp}. With FairCMS-I or FairCMS-II, the owner's computing overhead can be reduced to less than $0.2$ seconds, even if the number of users reaches hundreds. For another, FairCMS-I exhibits significantly higher cloud-side efficiency compared with \cite{zhang2018you,zheng2022towards,frattolillo2019multiparty}, and the cloud-side efficiency gain increases with the size of the media content. Thus, the comparison shown in Table \ref{tab:comparison} is validated.}

\begin{table}[htbp]
\centering
\setlength{\tabcolsep}{1.6mm}
\renewcommand\arraystretch{1.7}
 \scriptsize
\caption{{\color{black}The time (in second) required to encrypt a video with different frames using the single-value alteration encryption adopted by FairCMS-I}}
\begin{tabular}{c|c|c|c|c|c|c|c|c}
\hline
 {\color{black}Number of frames}  & {\color{black}600} & {\color{black}1200} & {\color{black}1800} & {\color{black}2400} & {\color{black}3000} & {\color{black}3600} & {\color{black}4200} & {\color{black}4800} \\
 \hline
{\color{black}Encryption} & {\color{black}129} & {\color{black}306} & {\color{black}461} & {\color{black}645} & {\color{black}809} & {\color{black}1012} & {\color{black}1182} & {\color{black}1406} \\ 
   \hline
{\color{black}Decryption} & {\color{black}136} & {\color{black}275} & {\color{black}484} & {\color{black}636} & {\color{black}852} & {\color{black}1066} & {\color{black}1283} & {\color{black}1455} \\
  \hline
\end{tabular}
\vspace{-2pt}
\label{my-Table5}
\end{table}

\begin{figure}[ht]
\centering
\includegraphics[scale=0.43]{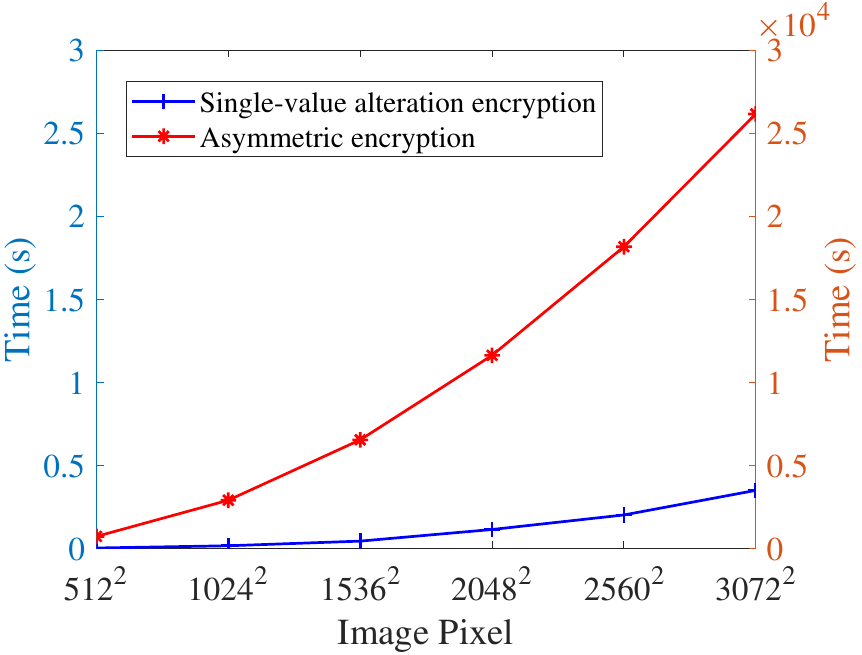}
\caption{{\color{black}Efficiency comparison between the single-value alteration encryption and the asymmetric encryption (i.e., the second-level encryption of the lifted-ElGamal based PRE).} The blue and red lines correspond to the Y-axis on the left and right, respectively.}
\label{fig14}
\vspace{-2pt}
\end{figure}

\begin{figure*}[ht]
\centering
\begin{minipage}{0.32\linewidth}
\vspace{3pt}
\centerline{\includegraphics[width=\textwidth]{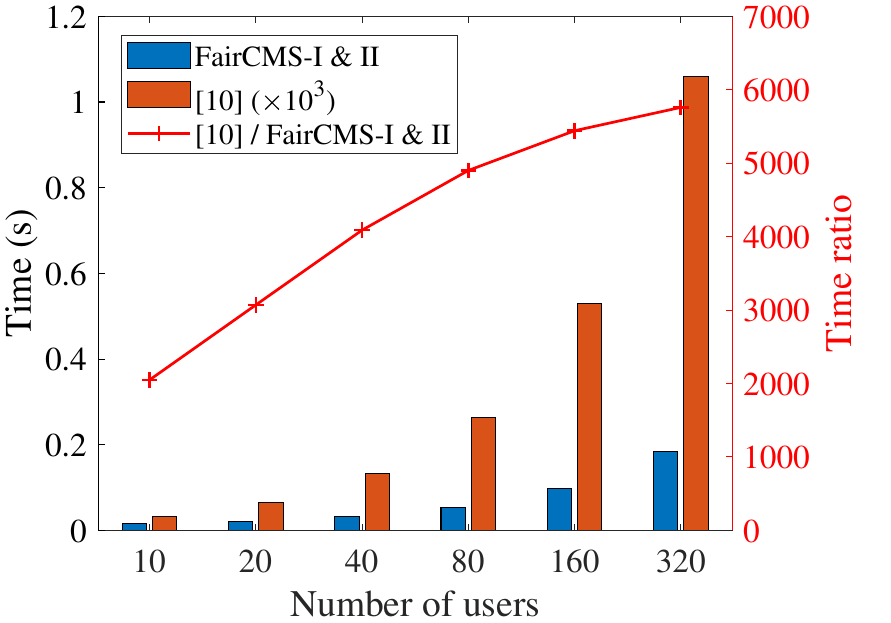}}
\centerline{(a)}
\end{minipage}
\begin{minipage}{0.315\linewidth}
\vspace{3pt}
\centerline{\includegraphics[width=\textwidth]{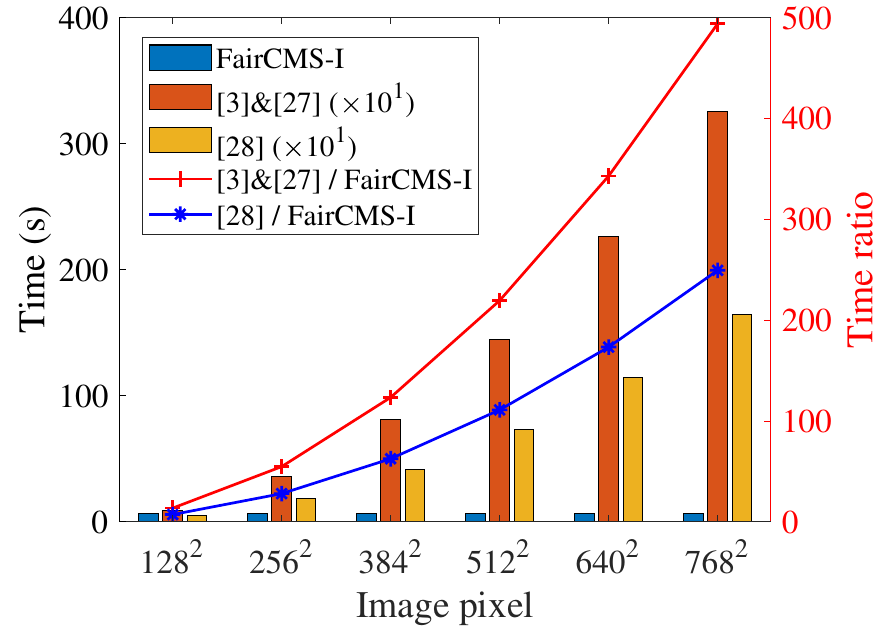}}
\centerline{(b)}
\end{minipage}
\begin{minipage}{0.335\linewidth}
\vspace{3pt}
\centerline{\includegraphics[width=\textwidth]{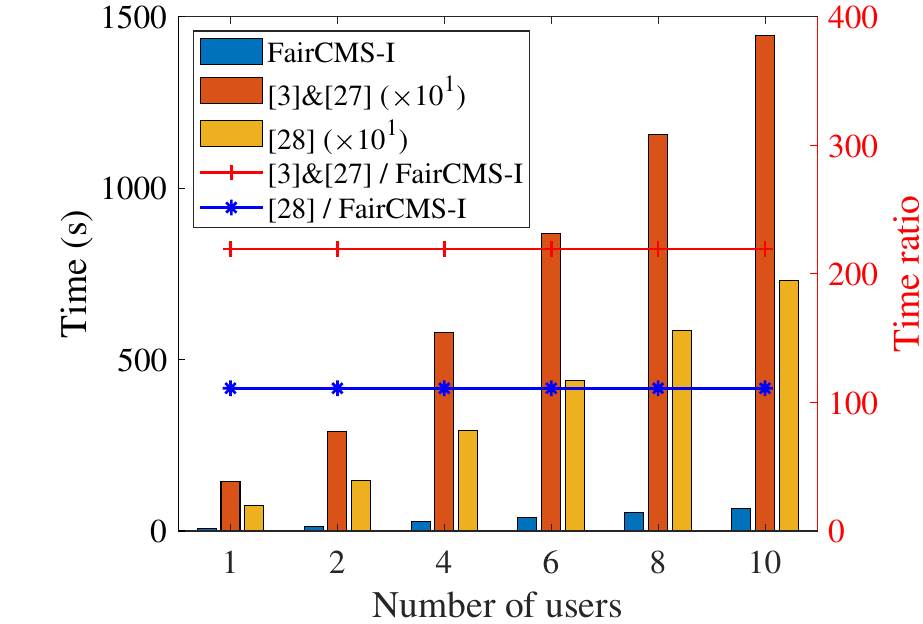}}
\centerline{(c)}
\end{minipage}
\caption{{\color{black}Comparison of computational efficiency with existing works. The bars and polyline correspond to the left and right Y-axes, respectively. (a) Efficiency comparison with the AFP scheme \cite{bianchi2014ttp} on the owner side. (b) Efficiency comparison with \cite{zhang2018you,zheng2022towards,frattolillo2019multiparty} on the cloud side under different image pixels. (c) Efficiency comparison with \cite{zhang2018you,zheng2022towards,frattolillo2019multiparty} on the cloud side under different number of users. The time consumed by \cite{bianchi2014ttp} and \cite{zhang2018you,zheng2022towards,frattolillo2019multiparty} is 1,000 times and 10 times the reading on the Y-axis, respectively.}}
\label{fig15}
\vspace{-2pt}
\end{figure*}


\section{Conclusion}
{\color{black}This paper solves the three problems faced by cloud media sharing and proposes two schemes FairCMS-I and FairCMS-II. FairCMS-I gives a method to transfer the management of LUTs to the cloud, enabling the calculation of each user's D-LUT in the ciphertext domain and its subsequent distribution. However, utilizing the single-value alteration method for masking the original media content does not achieve the IND-CPA security. Then FairCMS-II offers an enhanced privacy solution by replacing the encryption method with the lifted-ElGamal based PRE scheme, albeit at the cost of increased cloud overhead. Notably, both FairCMS-I and FairCMS-II fulfill scalability and owner-side efficiency requirements. In summary, the two proposed schemes can facilitate the media sharing of owners, while simultaneously ensuring the joint protection of copyright and users' rights, ultimately promoting the sustainable growth of the media sharing industry.}
  
\bibliographystyle{ieeetr}
\bibliography{references}

\begin{thebibliography}{10}

\bibitem{hu2020cover}
R.~Hu and S.~Xiang, ``Cover-lossless robust image watermarking against
  geometric deformations,'' {\em IEEE Trans. Image Process.}, vol.~30,
  pp.~318--331, 2021.

\bibitem{wang2023data}
Z.~Wang, O.~Byrnes, H.~Wang, R.~Sun, C.~Ma, H.~Chen, Q.~Wu, and M.~Xue, ``Data
  hiding with deep learning: A survey unifying digital watermarking and
  steganography,'' {\em IEEE Trans. Comput. Soc. Syst.}, 2023, in press,
  doi:10.1109/TCSS.2023.3268950.

\bibitem{zhang2018you}
L.~Y. Zhang, Y.~Zheng, J.~Weng, C.~Wang, Z.~Shan, and K.~Ren, ``You can access
  but you cannot leak: Defending against illegal content redistribution in
  encrypted cloud media center,'' {\em IEEE Trans. Depend. Secure Comput.},
  vol.~17, no.~6, pp.~1218--1231, 2020.

\bibitem{dong2020watermarking}
X.~Dong, W.~Zhang, M.~Shah, B.~Wang, and N.~Yu, ``Watermarking-based secure
  plaintext image protocols for storage, show, deletion and retrieval in the
  cloud,'' {\em IEEE Trans. Services Comput.}, vol.~15, no.~3, pp.~1678--1692,
  2022.

\bibitem{bobba2010attribute}
R.~Bobba, O.~Fatemieh, F.~Khan, A.~Khan, C.~A. Gunter, H.~Khurana, and
  M.~Prabhakaran, ``Attribute-based messaging: Access control and
  confidentiality,'' {\em ACM Trans. Inf. System Secur.}, vol.~13, no.~4,
  pp.~1--35, 2010.

\bibitem{wei2018fractal}
W.~Wei, S.~Liu, W.~Li, and D.~Du, ``Fractal intelligent privacy protection in
  online social network using attribute-based encryption schemes,'' {\em IEEE
  Trans. Comput. Soc. Syst.}, vol.~5, no.~3, pp.~736--747, 2018.

\bibitem{xu2012cl}
L.~Xu, X.~Wu, and X.~Zhang, ``{CL-PRE}: {A} certificateless proxy re-encryption
  scheme for secure data sharing with public cloud,'' in {\em Proc. 7th ACM
  Symp. Inf. Comput. Commun. Secur.}, pp.~87--88, 2012.

\bibitem{ateniese2006improved}
G.~Ateniese, K.~Fu, M.~Green, and S.~Hohenberger, ``Improved proxy
  re-encryption schemes with applications to secure distributed storage,'' {\em
  ACM Trans. Inf. System Secur.}, vol.~9, no.~1, pp.~1--30, 2006.

\bibitem{pfitzmann1996asymmetric}
B.~Pfitzmann and M.~Schunter, ``Asymmetric fingerprinting,'' in {\em Proc.
  {EUROCRYPT}}, vol.~96, pp.~84--95, Springer, 1996.

\bibitem{bianchi2014ttp}
T.~Bianchi and A.~Piva, ``{TTP}-free asymmetric fingerprinting based on client
  side embedding,'' {\em IEEE Trans. Inf. Forensics Secur.}, vol.~9, no.~10,
  pp.~1557--1568, 2014.

\bibitem{memon2001buyer}
N.~Memon and P.~W. Wong, ``A buyer-seller watermarking protocol,'' {\em IEEE
  Trans. Image Process.}, vol.~10, no.~4, pp.~643--649, 2001.

\bibitem{lei2004efficient}
C.-L. Lei, P.-L. Yu, P.-L. Tsai, and M.-H. Chan, ``An efficient and anonymous
  buyer-seller watermarking protocol,'' {\em IEEE Trans. Image Process.},
  vol.~13, no.~12, pp.~1618--1626, 2004.

\bibitem{rial2010provably}
A.~Rial, M.~Deng, T.~Bianchi, A.~Piva, and B.~Preneel, ``A provably secure
  anonymous buyer--seller watermarking protocol,'' {\em IEEE Trans. Inf.
  Forensics Secur.}, vol.~5, no.~4, pp.~920--931, 2010.

\bibitem{bianchi2013secure}
T.~Bianchi and A.~Piva, ``Secure watermarking for multimedia content
  protection: A review of its benefits and open issues,'' {\em IEEE Signal
  Process. Mag.}, vol.~30, no.~2, pp.~87--96, 2013.

\bibitem{qin2016survey}
Z.~Qin, H.~Xiong, S.~Wu, and J.~Batamuliza, ``A survey of proxy re-encryption
  for secure data sharing in cloud computing,'' {\em IEEE Trans. Services
  Comput.}, 2016, in press, doi:10.1109/TSC.2016.2551238.

\bibitem{yu2019file}
B.~Yu, C.~Zhang, and W.~Li, ``File matching based on secure authentication and
  proxy homomorphic re-encryption,'' in {\em Proc. 11th Int. Conf. Machine
  Learning Comput.}, pp.~472--476, ACM, 2019.

\bibitem{samanthula2015secure}
B.~K. Samanthula, Y.~Elmehdwi, G.~Howser, and S.~Madria, ``A secure data
  sharing and query processing framework via federation of cloud computing,''
  {\em Inf. Systems}, vol.~48, pp.~196--212, 2015.

\bibitem{gao2019cloud}
C.-z. Gao, Q.~Cheng, X.~Li, and S.-b. Xia, ``Cloud-assisted privacy-preserving
  profile-matching scheme under multiple keys in mobile social network,'' {\em
  Cluster Comput.}, vol.~22, no.~1, pp.~1655--1663, 2019.

\bibitem{shafagh2017secure}
H.~Shafagh, A.~Hithnawi, L.~Burkhalter, P.~Fischli, and S.~Duquennoy, ``Secure
  sharing of partially homomorphic encrypted {IOT} data,'' in {\em Proc. 15th
  ACM Conf. Embedded Netw. Sensor Systems}, pp.~1--14, ACM, 2017.

\bibitem{derler2017homomorphic}
D.~Derler, S.~Ramacher, and D.~Slamanig, ``Homomorphic proxy re-authenticators
  and applications to verifiable multi-user data aggregation,'' in {\em Proc.
  Int. Conf. Financial Cryptography Data Secur.}, pp.~124--142, Springer, 2017.

\bibitem{wu2013attribute}
Y.~Wu, Z.~Wei, and R.~H. Deng, ``Attribute-based access to scalable media in
  cloud-assisted content sharing networks,'' {\em IEEE Trans. Multimedia},
  vol.~15, no.~4, pp.~778--788, 2013.

\bibitem{polyakov2017fast}
Y.~Polyakov, K.~Rohloff, G.~Sahu, and V.~Vaikuntanathan, ``Fast proxy
  re-encryption for publish/subscribe systems,'' {\em ACM Trans. Privacy
  Secur.}, vol.~20, no.~4, pp.~1--31, 2017.

\bibitem{liang2014dfa}
K.~Liang, M.~H. Au, J.~K. Liu, W.~Susilo, D.~S. Wong, G.~Yang, T.~V.~X. Phuong,
  and Q.~Xie, ``A {DFA}-based functional proxy re-encryption scheme for secure
  public cloud data sharing,'' {\em IEEE Trans. Inf. Forensics Secur.}, vol.~9,
  no.~10, pp.~1667--1680, 2014.

\bibitem{li2018multi}
J.~Li, X.~Chen, S.~S. Chow, Q.~Huang, D.~S. Wong, and Z.~Liu, ``Multi-authority
  fine-grained access control with accountability and its application in
  cloud,'' {\em J. Netw. Computer Appl.}, vol.~112, pp.~89--96, 2018.

\bibitem{poh2008efficient}
G.~S. Poh and K.~M. Martin, ``An efficient buyer-seller watermarking protocol
  based on {Chameleon} encryption,'' in {\em Proc. Int. Work. Digital
  Watermarking}, pp.~433--447, Springer, 2008.

\bibitem{xia2016privacy}
Z.~Xia, X.~Wang, L.~Zhang, Z.~Qin, X.~Sun, and K.~Ren, ``A privacy-preserving
  and copy-deterrence content-based image retrieval scheme in cloud
  computing,'' {\em IEEE Trans. Inf. Forensics Secur.}, vol.~11, no.~11,
  pp.~2594--2608, 2016.

\bibitem{zheng2022towards}
T.~Zheng, Y.~Luo, T.~Zhou, and Z.~Cai, ``Towards differential access control
  and privacy-preserving for secure media data sharing in the cloud,'' {\em
  Computers Secur.}, vol.~113, p.~102553, 2022.

\bibitem{frattolillo2019multiparty}
F.~Frattolillo, ``A multiparty watermarking protocol for cloud environments,''
  {\em J. Inf. Secur. Appl.}, vol.~47, pp.~246--257, 2019.

\bibitem{seo2014efficient}
S.-H. Seo, M.~Nabeel, X.~Ding, and E.~Bertino, ``An efficient certificateless
  encryption for secure data sharing in public clouds,'' {\em IEEE Trans.
  Knowl. Data Eng.}, vol.~26, no.~9, pp.~2107--2119, 2014.

\bibitem{celik2008lookup}
M.~U. Celik, A.~N. Lemma, S.~Katzenbeisser, and M.~van~der Veen,
  ``Lookup-table-based secure client-side embedding for spread-spectrum
  watermarks,'' {\em IEEE Trans. Inf. Forensics Secur.}, vol.~3, no.~3,
  pp.~475--487, 2008.

\bibitem{celik2007secure}
M.~U. Celik, A.~N. Lemma, S.~Katzenbeisser, and M.~van~der Veen, ``Secure
  embedding of spread spectrum watermarks using look-up-tables,'' in {\em Proc.
  IEEE Int. Conf. Acoust., Speech, Signal Process. (ICASSP)}, vol.~2,
  pp.~153--156, 2007.

\bibitem{Marshall2003Coding}
T.~Marshall, ``Coding of real-number sequences for error correction: A digital
  signal processing problem,'' {\em IEEE J. Sel. Areas Commun.}, vol.~2, no.~2,
  pp.~381--392, 2003.

\bibitem{wang2003complex}
Z.~Wang and G.~B. Giannakis, ``Complex-field coding for {OFDM} over fading
  wireless channels,'' {\em IEEE Trans. Inf. Theory}, vol.~49, no.~3,
  pp.~707--720, 2003.

\bibitem{CHEN2022}
Y.~Chen, B.~Wang, H.~Jiang, P.~Duan, Y.~Ping, and Z.~Hong, ``{PEPFL}: A
  framework for a practical and efficient privacy-preserving federated
  learning,'' {\em Digital Commun. Netw.}, 2022, in press, doi:
  10.1016/j.dcan.2022.05.019.

\bibitem{boneh2003identity}
D.~Boneh and M.~Franklin, ``Identity-based encryption from the weil pairing,''
  {\em {SIAM} J. Comput.}, vol.~32, no.~3, pp.~586--615, 2003.

\bibitem{2011Efficient}
D.~Evans, Y.~Huang, J.~Katz, and L.~Malka, ``Efficient privacy-preserving
  biometric identification,'' in {\em Proc. Netw. Distri. System Secur. Symp.
  (NDSS)}, vol.~68, pp.~90--98, 2011.

\bibitem{xiao2023fingerchain}
X.~Xiao, Y.~Zhang, Y.~Zhu, P.~Hu, and X.~Cao, ``Fingerchain: Copyrighted
  multi-owner media sharing by introducing asymmetric fingerprinting into
  blockchain,'' {\em IEEE Trans. Netw. Serv. Manage.}, vol.~20, no.~3,
  pp.~2869--2885, 2023.

\bibitem{xiao2022preview}
X.~Xiao, X.~Ye, Y.~Zhang, W.~Wen, and X.~Zhang, ``Preview-supported copyright
  image sharing,'' {\em J. Electronics $\&$ Inf. Technol.}, vol.~45, pp.~1--10,
  2022.

\bibitem{bianchi2015anticollusion}
T.~Bianchi, A.~Piva, and D.~Shullani, ``Anticollusion solutions for asymmetric
  fingerprinting protocols based on client side embedding,'' {\em EURASIP J.
  Inf. Secur.}, vol.~2015, no.~1, pp.~1--17, 2015.

\bibitem{frattolillo2016buyer}
F.~Frattolillo, ``A buyer-friendly and mediated watermarking protocol for web
  context,'' {\em ACM Trans. Web}, vol.~10, no.~2, pp.~1--28, 2016.

\bibitem{huan2021exploring}
W.~Huan, S.~Li, Z.~Qian, and X.~Zhang, ``Exploring stable coefficients on joint
  sub-bands for robust video watermarking in dt cwt domain,'' {\em IEEE Trans.
  Circuits Syst. Video Technol.}, vol.~32, no.~4, pp.~1955--1965, 2021.

\bibitem{salah2021fourier}
E.~Salah, K.~Amine, K.~Redouane, and K.~Fares, ``A {Fourier} transform based
  audio watermarking algorithm,'' {\em Appl. Acoustics}, vol.~172, p.~107652,
  2021.

\end{thebibliography}

\begin{IEEEbiography}[{\includegraphics[width=1in,height=1.25in,clip,keepaspectratio]{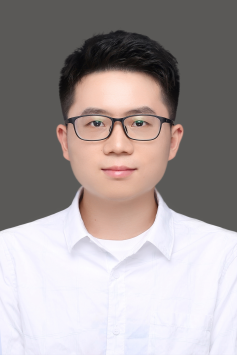}}]{Xiangli Xiao}
received the B.E. degree from the College of Electronic and Information Engineering, Southwest University, Chongqing, China, in Jun. 2020. He is currently pursuing the Ph.D. degree in the College of Computer Science and Technology, Nanjing University of Aeronautics and Astronautics, Nanjing, China. His current research interests include multimedia security, digital watermarking, blockchain, and cloud computing security.
\end{IEEEbiography}

\begin{IEEEbiography}[{\includegraphics[width=1in,height=1.25in,clip,keepaspectratio]{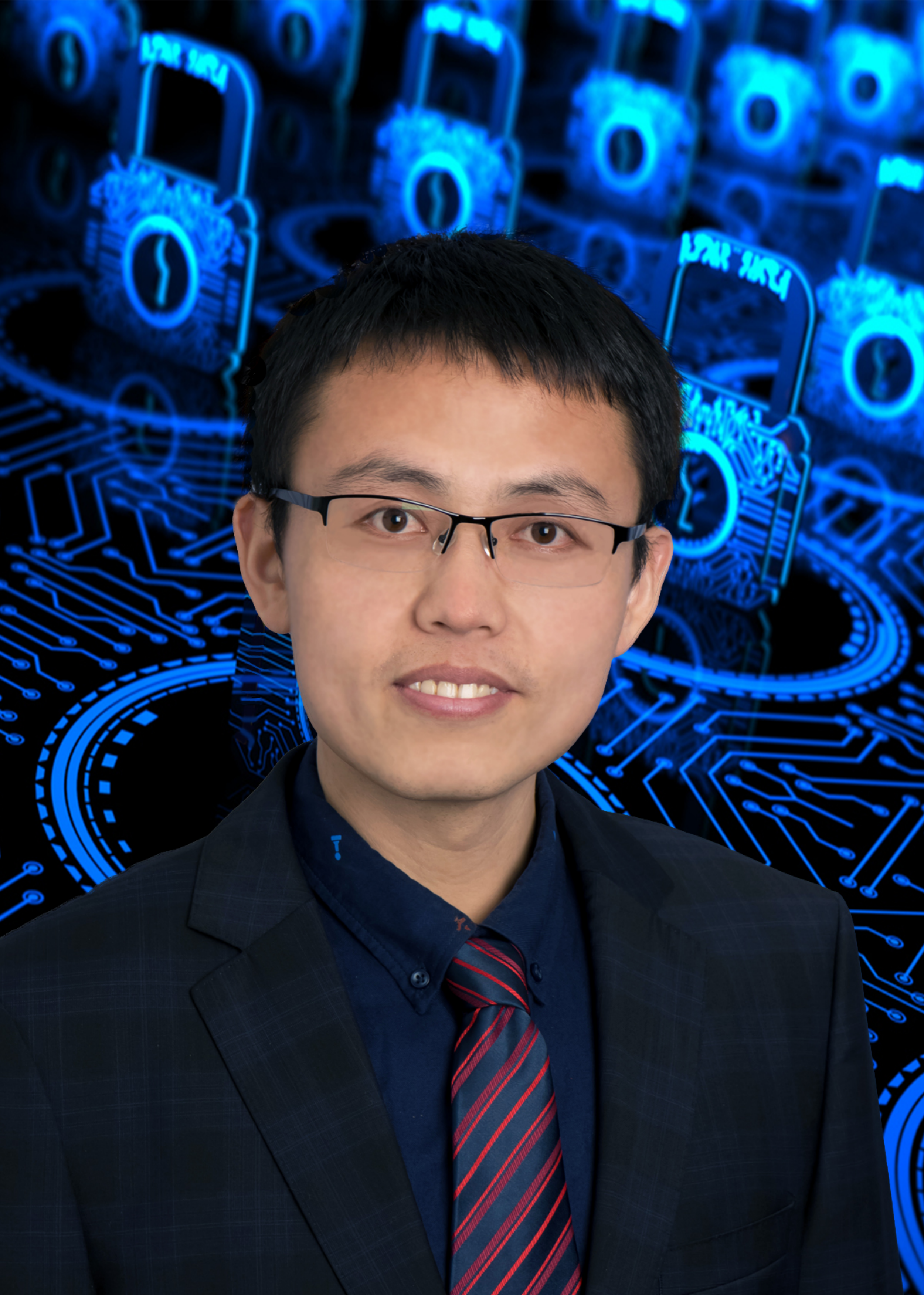}}]{Yushu Zhang}
(Senior Member, IEEE) received the B.S. degree from the School of Science, North University of China, Taiyuan, China, in 2010, and the Ph.D. degree from the College of Computer Science, Chongqing University, Chongqing, China, in 2014. He held various research positions with City University of Hong Kong, Southwest University, University of Macau, and Deakin University. He is currently a Professor with the College of Computer Science and Technology, Nanjing University of Aeronautics and Astronautics, Nanjing, China. He is an Associate Editor of Information Sciences, Journal of King Saud University-Computer and Information Sciences, and Signal Processing. His research interests include multimedia security, blockchain, and artificial intelligence. He has co-authored more than 200 refereed journal articles and conference papers in these areas.
\end{IEEEbiography}

\begin{IEEEbiography}[{\includegraphics[width=1in,height=1.25in,clip,keepaspectratio]{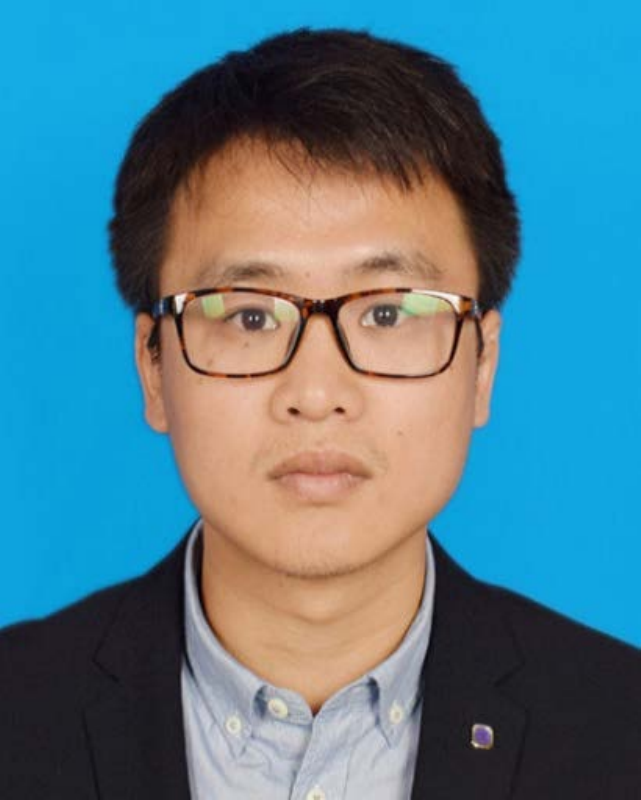}}]{Leo Yu Zhang}
(Member, IEEE) received the Ph.D. degree from the City University of Hong Kong, Hong Kong, China, in 2016. He is currently a Senior Lecturer with the School of Information and Communication Technology, Griffith University, QLD, Australia. He used to be a faculty member at the School of Information Technology, Deakin University, from 2018 to 2023. He held various research positions with the City University of Hong Kong, the University of Macau, the University of Ferrara, and the University of Bologna. He is an Associate Editor of IEEE Transactions on Dependable and Secure Computing. His current research focuses on trustworthy AI and applied cryptography, and he has published over 100 articles in refereed journals and conferences, such as TIFS, TDSC, Oakland, AsiaCCS, NeurIPS, CVPR, ICCV, IJCAI, AAAI, etc.
\end{IEEEbiography}

\begin{IEEEbiography}[{\includegraphics[width=1in,height=1.25in,clip,keepaspectratio]{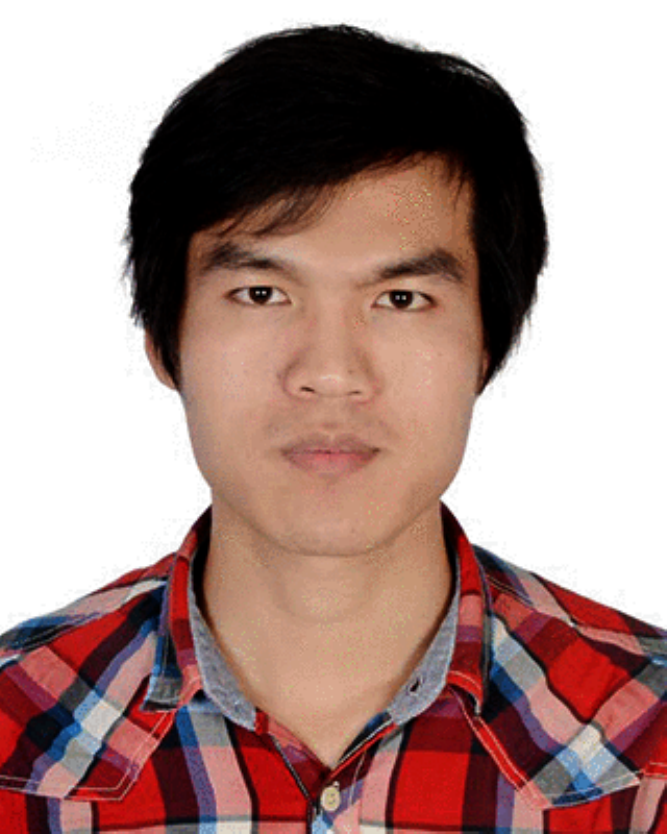}}]{Zhongyun Hua}
(Senior Member, IEEE) received the B.S. degree from Chongqing University, Chong-qing, China, in 2011, followed by the M.S. and Ph.D. degrees from the University of Macau, Macau, China, in 2013 and 2016, respectively. He is currently an Associate Professor with the School of Computer Science and Technology, Harbin Institute of Technology, Shenzhen, China. His works have appeared in prestigious venues such as IEEE Transactions on Dependable and Secure Computing, IEEE Transactions on Image Processing, IEEE Transactions on Signal Processing, and ACM Multimedia. He has been recognized as a ‘Highly Cited Researcher 2022’. His current research interests are focused on chaotic system, multimedia security, and secure cloud computing. He has published about seventy papers on the subject, receiving more than 5800 citations.
\end{IEEEbiography}

\begin{IEEEbiography}[{\includegraphics[width=1in,height=1.25in,clip,keepaspectratio]{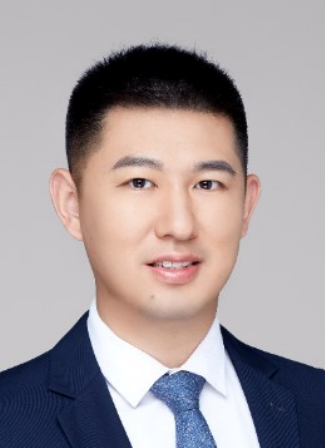}}]{Zhe Liu}
(Senior Member, IEEE) received the B.S. and M.S. degrees from Shandong University, China, in 2008 and 2011, respectively, and the Ph.D. degree from the Laboratory of Algorithmics, Cryptology and Security, University of Luxembourg, Luxembourg, in 2015. He is currently a Professor with the College of Computer Science and Technology, Nanjing University of Aeronautics and Astronautics, Nanjing, China. He has co-authored over 100 peer-reviewed journal articles and conference papers. His research interests include security, privacy, and cryptographic solutions for the Internet of Things. He has served as a program committee member for more than 50 international conferences. He was a recipient of the prestigious FNR Awards-Outstanding Ph.D. Thesis Award in 2016, the ACM CHINA SIGSAC Rising Star Award in 2017, as well as the DAMO Academy Young Fellow in 2019.
\end{IEEEbiography}

\begin{IEEEbiography}[{\includegraphics[width=1in,height=1.25in,clip,keepaspectratio]{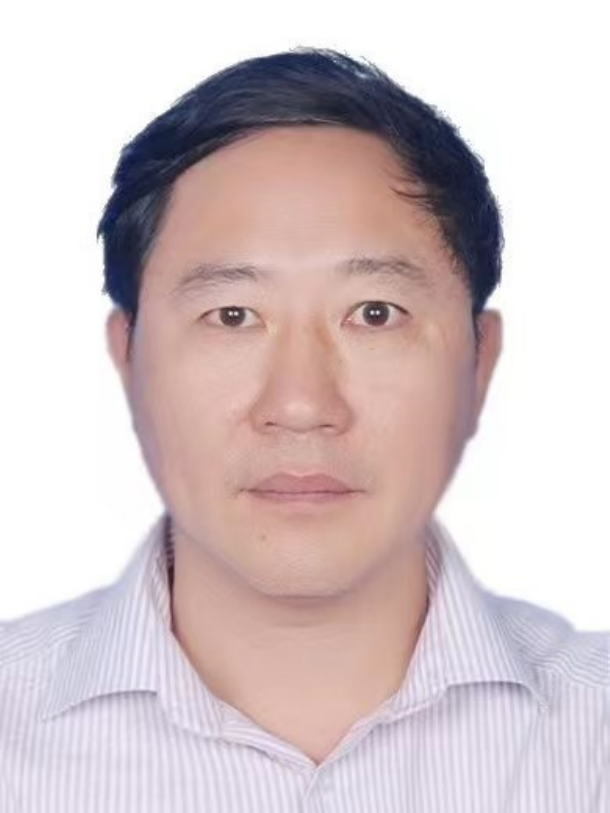}}]{Jiwu Huang}
(Fellow, IEEE) received the B.S. degree from Xidian University, Xi'an, China, the M.S. degree from Tsinghua University, Beijing, China, and the Ph.D. degree from the Institute of Automation, Chinese Academy of Sciences, Beijing, in 1982, 1987, and 1998, respectively. He is currently a Professor with the College of Electronics and Information Engineering, Shenzhen University, Shenzhen, China. His current research interests include multimedia forensics and security. He serves as a member of the IEEE CASS Multimedia Systems and Applications Technical Committee and the IEEE SPS Information Forensics and Security Technical Committee. He is an Associate Editor for the IEEE Transactions on Information Forensics and Security, the LNCS Transactions on Data Hiding and Multimedia Security (Springer), and the EURASIP Journal on Information Security (Hindawi).
\end{IEEEbiography}

\end{document}